\newif\ifanonymous
\newif\ifdraft

\draftfalse
\InputIfFileExists{localflags}{}

\PassOptionsToPackage{usenames,dvipsnames}{xcolor}
\documentclass{IEEEtran}

\ifdraft
\fi

\usepackage{versions}
\usepackage{xr}

\excludeversion{conf}
\excludeversion{revision}
\includeversion{full}

\usepackage[skins]{tcolorbox}
\usepackage{listings}
\usepackage{enumitem}
\usepackage{stmaryrd}
\usepackage{regexpatch}
\usepackage{ebproof}
\usepackage{scalerel}
\usepackage{adjustbox}
\usepackage{xspace}
\usepackage{cryptocode}
\usepackage{graphicx}
\usepackage[]{multirow}
\usepackage{float}
\usepackage{hyperref}
\usepackage{cleveref}
\usepackage{booktabs} 
\usepackage{amsthm}
\usepackage{amssymb}
\usepackage{mathrsfs}
\usepackage{subcaption}
\usepackage{cite}
\makeatletter
\xpatchcmd{\@todo}{\setkeys{todonotes}{#1}}{\setkeys{todonotes}{#1}}{}{}
\makeatother

\usepackage{centernot}
\usepackage{pifont}

\newcommand{\OK}{\ding{51}}
\newcommand{\NO}{\ding{55}}

\newcommand{\newtext}[1]{#1}

\usepackage{todonotes}

\newcommand{\mathcmd}[1]{{\normalfont\ensuremath{#1}}\xspace}

\newcommand{\mathlabel}[1]{\mathcmd{\textsf{#1}}}

\newcommand{\textop}[1]{\relax\ifmmode\mathop{\text{#1}}\else\text{#1}\fi}

\newcommand{\set}[1]{\{#1\}}

\newcommand{\ledot}{\mathrel{\ooalign{\hss\raise.200ex\hbox{$\cdot$}\hss\cr$\le$}}}
\newcommand{\gedot}{\mathrel{\ooalign{\hss\raise.200ex\hbox{$\cdot$}\hss\cr$\ge$}}}

\newcommand{\tuple}[1]{( #1 )}

\renewcommand{\iff}{\Leftrightarrow}

\newcommand{\calC}{\ensuremath{\mathcal{C}}\xspace}

\newcommand{\calF}{\ensuremath{\mathcal{F}}\xspace}

\newcommand\calT{\ensuremath{\mathcal{T}}\xspace}

\let\orgautoref\autoref
\renewcommand{\autoref}
{\def\equationautorefname{Eq.\!}%
	\def\figureautorefname{Fig.\!}%
	\def\subfigureautorefname{Fig.\!}%
	\def\Itemautorefname{Item}%
	\def\tableautorefname{Table}%
	\def\algorithmautorefname{Algorithm}%
	\def\paragraphautorefname{Para.\!}%
	\def\sectionautorefname{Sec.\!}%
	\def\subsectionautorefname{Sec.\!}%
	\def\subsubsectionautorefname{Sec.\!}%
	\def\chapterautorefname{Chapter}%
	\def\partautorefname{Part}%
	\def\goalautorefname{Goal}%
	\def\reqautorefname{Req.\!}%
	\def\adviceautorefname{Rule}%
	\def\parameterautorefname{Param.\!}%
	\def\definitionautorefname{Def.\!}%
	\def\definitionsautorefname{Def.\!}%
	\def\propertyautorefname{Property}%
	\def\lemmaautorefname{Lemma}%
	\def\theoremautorefname{Thm.\!}%
	\orgautoref}

\newcommand{\infruleref}[2]{\hyperref[#1]{#2}}
\newcommand{\fullversionref}{cryptobap-parallel-full}

\newcommand{\appendixorextended}[1]{%
\processifversion{conf}{extended
version~\cite[Appendix~\ref{F-#1}]{\fullversionref}}%
\processifversion{revision}{extended
	version~\cite[Appendix~\ref{F-#1}]{\fullversionref}}%
\processifversion{full}{Appendix~\ref{#1}}%
}

\newcommand{\theappendixorextended}[1]{%
\processifversion{conf}{the extended version~\cite[Appendix~\ref{F-#1}]{\fullversionref}}%
\processifversion{revision}{the extended version~\cite[Appendix~\ref{F-#1}]{\fullversionref}}%
\processifversion{full}{Appendix~\ref{#1}}%
        }

\newcommand{\appendixorfull}[1]{%
	\processifversion{conf}{full
		version~\cite[Appendix~\ref{F-#1}]{\fullversionref}}%
	\processifversion{revision}{full
		version~\cite[Appendix~\ref{F-#1}]{\fullversionref}}%
	\processifversion{full}{Appendix~\ref{#1}}%
}

\newcommand{\theappendixorfull}[1]{%
	\processifversion{conf}{the full version~\cite[Appendix~\ref{F-#1}]{\fullversionref}}%
	\processifversion{revision}{the full version~\cite[Appendix~\ref{F-#1}]{\fullversionref}}%
	\processifversion{full}{Appendix~\ref{#1}}%
}

\newcommand{\refappendixorfull}[1]{%
	\processifversion{conf}{\cite[\autoref{F-#1}]{\fullversionref}}%
	\processifversion{revision}{\cite[\autoref{F-#1}]{\fullversionref}}%
	\processifversion{full}{\autoref{#1}}%
}

\newcommand{\revised}[2]{%
	\processifversion{conf}{#2}%
	\processifversion{full}{#2}%
	\processifversion{revision}{\marginnote{\hspace{-0.5em}\textcolor{blue}{#1}}
		{\textcolor{red}{#2}}}%
}

\newcommand{\revisedtxt}[1]{%
	\processifversion{conf}{#1}%
	\processifversion{full}{#1}%
	\processifversion{revision}{\textcolor{red}{#1}}%
}

\newcommand{\newrevised}[1]{%
	\processifversion{conf}{#1}%
	\processifversion{full}{#1}%
	\processifversion{revision}{#1}%
}
\newtheorem{definition}{Definition}
\newcommand\definitionautorefname{Definition}
\newtheorem{lemma}{Lemma}
\newcommand\lemmaautorefname{Lemma}
 \newtheorem{theorem}{Theorem}
\newtheorem{example}{Example}
\newtheorem{corollary}{Corollary}

\numberwithin{subcase}{mycase}

\newcommand{\mi}[1]{\ensuremath{\mathit{#1}}}

\newcommand{\mtt}[1]{\ensuremath{\mathtt{#1}}}
\newcommand{\mf}[1]{\ensuremath{\mathbf{#1}}}

\newcommand{\ms}[1]{\ensuremath{\mathsf{#1}}}

\newcommand{\neutcol}[0]{black}
\newcommand{\rblu}[0]{RoyalBlue}
\newcommand{\redo}[0]{RedOrange}
\newcommand{\plm}[0]{Plum}

\newcommand{\commoncol}[0]{black}    

\newcommand{\col}[2]{\ensuremath{{\color{#1}{#2}}}}

\newcommand{\fstlangcol}[1]{\mf{\col{\rblu}{#1}}}
\newcommand{\sndlangcol}[1]{\ms{\col{\redo}{#1}}}
\newcommand{\trdlangcol}[1]{\mtt{\col{\plm}{#1}}}

\newcommand{\bl}[1]{\col{\neutcol }{#1}}
\newcommand{\com}[1]{\mi{\col{\commoncol }{#1}}}

\newcommand{\colone}[0]{\fstlangcol{1}}
\newcommand{\coltwo}[0]{\sndlangcol{2}}

\newcommand{\Sapic}{\sndlangcol{\textsc{Sapic}^+}\xspace}
\newcommand{\SapicOriginal}{\sndlangcol{\textsc{Sapic}}\xspace}
\newcommand{\SapicOld}{\sndlangcol{\textsc{Sapic}^-}\xspace}
\newcommand{\ProVerif}{\textsc{ProVerif}\xspace}

\newcommand{\Tamarin}{\textsc{Tamarin}\xspace}
\newcommand{\DeepSec}{\textsc{DeepSec}\xspace}
\newcommand{\Squirrel}{\textsc{Squirrel}\xspace}

\newcommand{\configSpace}{\calC}
\newcommand{\config}{c}
\newcommand{\stateSpace}{\tilde \calC}
\newcommand{\symbset}{\Sigma}
\newcommand{\symbSpace}{\mathcal{E}}
\newcommand{\predset}{\text{\ensuremath{\Pi}}}
\newcommand{\pred}{\varphi}
\newcommand{\predcomb}{\pred_{\colonetwo}}
\newcommand{\predSpace}{\mathcal{P}}

\newcommand{\dedrel}{\vdash}

\newcommand{\xxrightarrowsgl}[2][]{{%
		\xrightarrow[]{\bl{#2}}\mathrel{}_\bl{#1}
}}
\newcommand{\nrightarrowsgl}[2][]{{%
	\centernot{\xrightarrow[]{\bl{#2}}\mathrel{}_\bl{#1}}
}}
\newcommand{\ntransrel}[4]{ #3 \nrightarrowsgl{#1}_{#2} #4}
\newcommand{\transrel}[4]{#3 \xxrightarrowsgl{#1}_{#2} #4}

\newcommand{\PtoA}[1]{\com{P2A}(#1)}
\newcommand{\AtoP}[1]{\com{A2P}(#1)}
\newcommand{\PtoAsym}{\com{P2A}}
\newcommand{\AtoPsym}{\com{A2P}}
\newcommand{\Syncfresh}[1]{\com{SFr}(#1)}
\newcommand{\Syncfreshsym}{\com{SFr}}
\newcommand{\Silent}[1]{\com{Silent}(#1)}
\newcommand{\Silentsym}{\com{Silent}}
\newcommand{\Alias}[2]{\sndlangcol{Alias}(#1,#2)}

\newcommand{\System}{S}
\newcommand{\SysFst}[1][]{\fstlangcol{\System{#1}}}
\newcommand{\SysSnd}[1][]{\sndlangcol{\System{#1}}}

\newcommand{\colonetwo}{\colone \coltwo}
\newcommand{\dedrelcomb}{\dedrel_{\colonetwo}}
\newcommand{\subone}[1]{\fstlangcol{#1_1}}
\newcommand{\subtwo}[1]{\sndlangcol{#1_2}}

\newcommand{\eqone}{\fstlangcol{\doteq}}
\newcommand{\eqtwo}{\sndlangcol{\doteq}}

\newcommand{\dycol}[1]{\sndlangcol{#1}}
\newcommand{\libcol}[1]{\trdlangcol{#1}}
\newcommand{\fstlangpred}[1]{\fstlangcol{\mathlabel{#1}}}

\newcommand{\pubnames}{\mathcal{N}_\texttt{pub}}
\newcommand{\privnames}{\mathcal{N}_\texttt{priv}}
\newcommand{\allnames}{\mathcal{N}}
\newcommand{\Vars}{\mathcal{V}}

\newcommand{\knowledgeset}{\dycol{\predset_{ _A}}}
\newcommand{\knowledge}[1]{\dycol{\mathcal{K}}(#1)}

\newcommand{\freshness}[1]{\dycol{\mathlabel{Fr}}(#1)}

\newcommand{\sbirpred}{\fstlangcol{\predset_{\sbiridx}}}

\newcommand{\libknowledgeset}{\libcol{\predset_{ _L}}}

\newcommand{\libfreshness}[1]{\libcol{\mathlabel{Fr}}(#1)}
\newcommand{\libstate}{\libcol{\epsilon}}

\newcommand{\fcall}[1]{\com{{FCall}}(#1)}
\newcommand{\fcallsym}{\com{{FCall}}}

\newcommand{\progEqual}[2]{#1 \ \fstlangcol{\doteq} \ #2}
\newcommand{\dyEqual}[2]{#1 \ \dycol{\doteq} \ #2}
\newcommand{\dyAliEq}[2]{#1 \ \dycol{\stackrel{\cdot}{\mapsto}} \ #2}
\newcommand{\callEq}[2]{#1 \ \trdlangcol{\stackrel{\cdot}{\mapsto}} \ #2}

\newcommand{\Notequal}[2]{#1 \not= #2}

\newcommand{\dydedrel}{\dycol{\vdash_{ _A}}}
\newcommand{\dydedrelfun}[2]{#1 \ \dydedrel #2}
\newcommand{\dystate}{\dycol{\epsilon}}

\newcommand{\comppredset}{\Pi_{\colone \coltwo}}

\newcommand{\compdedrel}{\vdash_{\colone \coltwo}}
\newcommand{\compdedrelfun}[2]{#1 \compdedrel #2}

\newcommand{\proj}[2]{ #2\!\downharpoonright_{_{#1\!}}}

\newcommand{\Trace}{\mathfrak{t}}
\newcommand{\Traces}{\mathfrak{T}}
\newcommand{\Tracesfun}[1]{\Traces(#1)}
\newcommand{\mixtrace}{\Trace_{\colone \coltwo}}
\newcommand{\mixtraces}{\Traces_{\colone \coltwo}}
\newcommand{\sbirmixtrace}{\Trace_{\colone \coltwo}^{\sbiridx}}
\newcommand{\conmixtrace}{\Trace_{\colone \coltwo}^{\conidx}}
\newcommand{\sbirmixtraces}{\Traces_{\colone \coltwo}^{\sbiridx}}
\newcommand{\conmixtraces}{\Traces_{\colone \coltwo}^{\conidx}}

\newcommand{\sbirTrace}{\fstlangcol{\Trace^{\sbiridx}}}
\newcommand{\sbirtraces}{\fstlangcol{\Traces^{\sbiridx}}}
\newcommand{\InterTraces}{\mathscr{T}}
\newcommand{\InterTrace}{\mathtt{t}}

\newcommand{\Parallel}[3][~]{#2 \parallel^{#1} #3}
\newcommand{\CParallel}[2]{#1 \parallel_{c} #2}
\newcommand{\SParallel}[3][~]{#2 \parallel^{#1}_{s} #3}
\newcommand{\interleaving}[2]{#1\interleave #2}

\newcommand{\interpret}{\iota}

\newcommand{\sapicidx}{\sndlangcol{\scaleto{sp}{3.5pt}}}

\newcommand{\sapicTrace}{\sndlangcol{\Trace^{\sapicidx}}}
\newcommand{\sapictraces}{\sndlangcol{\Traces^{\sapicidx}}}

\newcommand{\sapicevent}{\sndlangcol{\eventsym^\sapicidx}}

\newcommand{\attidx}{\sndlangcol{\scaleto{A}{3pt}}}
\newcommand{\libidx}{\trdlangcol{\scaleto{L}{3pt}}}
\newcommand{\dedrelcombbit}{\dedrelcomb^\mathsf{bit}} 
\newcommand{\dedrelcombSbirDY}{\dedrel_{\fstlangcol{\sbiridx} \sndlangcol{\attidx}}^\mathsf{bit'}}
\newcommand{\dedrelcombSapicDY}{\dedrel_{\sndlangcol{\sapicidx} \sndlangcol{\attidx}}^{\sbirtoiml{\mathsf{bit'}}}}
\newcommand{\DYLib}[1]{\textsc{DYlib}_{\bl{#1}}}

\newcommand{\dedrelAvsL}{\dedrel_{\trdlangcol{\libidx} \sndlangcol{\attidx}}^{\stackrel{\cdot}{\mapsto}}}

\newcommand{\silentevent}{\tau}
\newcommand{\interpreterapply}{\mathit{apply}}
\newcommand{\progAliEq}[2]{#1 \ \fstlangcol{\stackrel{\cdot}{\mapsto}} \ #2}

\newcommand{\constrans}[1]{\text{\textquotesingle{}}#1\text{\textquotesingle{}}}

\newcommand{\traceproj}[2]{ #2\!\downarrow_{_{#1\!}}}

\newcommand{\Funcs}{\calF}

\newcommand{\Term}{\calT}

\newcommand{\Terms}{\ensuremath{\mathcal{T}}\xspace}

\newcommand{\pout}{\sndlangcol{\mathsf{out}}\xspace}
\newcommand{\pin}{\sndlangcol{\mathsf{in}}\xspace}

\newcommand{\pelse}{\sndlangcol{\mathsf{else}}\xspace}
\newcommand{\plet}{\sndlangcol{\mathsf{let}}\xspace}
\newcommand{\pnew}{\sndlangcol{\mathsf{new}}\xspace}

\newcommand{\pevent}{\sndlangcol{\mathsf{event}}\xspace}

\newcommand{\pndc}[2]{#1 \ \sndlangcol{+} \ #2}

\newcommand{\theactualrule}[1]{\text{Please redefine the command
theactualrule.}}
\newcommand{\underscorethingy}[1]{\text{Please redefine the command
underscorethingy.}}

\newcommand{\Kevent}{\ensuremath{\sndlangcol{\mathsf{K}}}\xspace}
\newcommand{\functionsymbol}[1]{\ensuremath{\mathlabel{#1}}}

\newcommand{\senc}{\functionsymbol{senc}}
\newcommand{\sdec}{\functionsymbol{sdec}}

\newcommand{\birsymb}{\fstlangcol{BIR}}
\newcommand{\sbirsymb}{\fstlangcol{SBIR}}
\newcommand{\statesymb}{\mathit{s}}

\newcommand{\bnfconcat}{\!::\!}
\newcommand{\bnfsep}{\;|\;}
\newcommand{\bnfdef}{\;\;:=\;\;}

\newcommand{\var}[1]{\fstlangcol{#1}}

\newcommand{\event}[1]{\fstlangcol{Ev}(#1)}
\newcommand{\assign}[2]{\fstlangcol{Asn}(#1,#2)}
\newcommand{\Loopsym}{\fstlangcol{{Loop}}}

\newcommand{\birpc}{\fstlangcol{pc}}

\newcommand{\birval}{\fstlangcol{Bval}}
\newcommand{\birexp}{\fstlangcol{Bexp}}
\newcommand{\birprog}{\fstlangcol{P}}

\newcommand{\eventsym}{\alpha}

\newcommand{\birprobinitstates}[1]{\fstlangcol{S_{\mathit{init}}^{b,n}}}
\newcommand{\birprobinitstate}[1]{\fstlangcol{\statesymb_{\mathit{init}}^{b,n}}}

\newcommand{\conidx}{{\scaleto{c}{3.5pt}}}
\newcommand{\contraces}{\Traces^{\conidx}}
\newcommand{\contrace}{\Trace^{\conidx}}
\newcommand{\sbiridx}{{s}}

\newcommand{\sbirevent}{\fstlangcol{\eventsym^\sbiridx}}

\newcommand{\symtraces}{\Traces^{\sbiridx}}
\newcommand{\symtrace}{\Trace^{\sbiridx}}

\newcommand{\defeq}{\mathrel{\stackrel{\makebox[0pt]{\mbox{\normalfont\tiny def}}}{=}}}

\newcommand{\sbirprobinitstates}[1]{\fstlangcol{S_{\mathit{init}}^{s,n}}}
\newcommand{\sbirprobinitstate}[1]{\fstlangcol{\statesymb_{\mathit{init}}^{s,n}}}

\newcommand{\sbirtoiml}[1]{\llbracket #1 \rrbracket}
\newcommand{\tracetrans}[1]{\llparenthesis #1 \rrparenthesis}

\newcommand{\imlvals}{\mathit{BS}}

\newcommand{\eventSpace}{\mathbb{E}}

\newcommand*{\suchthat}{\;%
	\ifnum\currentgrouptype=16\middle|\else%
	\iftoggle{WithinBracMacro}{\middle|}{|}%
	\fi%
	\;}%

\newcommand{\CryptoBap}{{\sc CryptoBap}}

\newcommand{\tree}{\fstlangcol{T}}
\newcommand{\node}{\fstlangcol{node}}
\newcommand{\branchingnode}{\fstlangcol{Branch}}
\newcommand{\leafnode}{\fstlangcol{Leaf}}
\newcommand{\nodecond}{\var{e}}
\newcommand{\nodeevent}{\fstlangcol{ev}}

\newlength\shlength
\newcommand\xshlongvec[2][0]{\setlength\shlength{#1pt}%
	\stackengine{-5.6pt}{$#2$}{\smash{$\kern\shlength%
			\stackengine{7.55pt}{$\mathchar"017E$}%
			{\rule{\widthof{$#2$}}{.57pt}\kern.4pt}{O}{r}{F}{F}{L}\kern-\shlength$}}%
	{O}{c}{F}{T}{S}}

\newcommand{\rPubName}{\textsc{Pub}}
\newcommand{\rKZero}{\ensuremath{\textsc{K}_0}}
\newcommand{\rAppl}{\ensuremath{\textsc{App}}}
\newcommand{\rEq}{\ensuremath{\textsc{Eq}}}
\newcommand{\rSubst}{\ensuremath{\textsc{Subst}}}
\newcommand{\rAlSubst}{\ensuremath{\textsc{Al-Subst}}}

\newcommand{\rPA}{\ensuremath{\textsc{P2A}}}
\newcommand{\rAP}{\ensuremath{\textsc{A2P}}}
\newcommand{\rAlias}{\ensuremath{\textsc{Alias}}}
\newcommand{\rFresh}{\ensuremath{\textsc{Fr-A2L}}}
\newcommand{\rFCall}{\ensuremath{\textsc{FCall}}}
\newcommand{\rFreshSync}{\ensuremath{\textsc{Fr-L2A}}}
\newcommand{\rDed}{\ensuremath{\textsc{Ded}}}

\definecolor{mGreen}{rgb}{0,0.6,0}
\definecolor{mGray}{rgb}{0.627, 0.627, 0.627}
\definecolor{mPurple}{rgb}{0.58,0,0.82}
\definecolor{backgroundColour}{rgb}{0.95,0.95,0.92}

\lstdefinelanguage
[x64]{Assembler}     
{morekeywords={cbnz,ldr, cmp, ldp, eq, ge, mul, nop, add, lsr, eor, and, M, Z, N, V, ldtm, mov, ret,b, bl, ble, bne,%
		udiv, bics, ldrb, ror, sxtw, for, do, end,  procedure, or, adrp,  %
		xzr, wzr,x0,X1,x1,X2,x2,X3,x3,w4,x4,x5,x6,x7,x8,x9, w0, w1, w2, w3,w4,w5,w6,w7,w8,w9,%
		x10,x11,x12,x13,x14,x15,x16,x17,x18,x19,%
		x30,x31,sp},
	alsoletter={x\#}}[strings,comments,keywords] %

\lstdefinelanguage{mlang} {morekeywords={system,type,var,process,do,observable,configuration,if,then,else,array,init,while,function,repeat,fun,for,return,ret}, 
	sensitive=false,
	morestring=[b]", }

\lstdefinestyle{asmstyle}{
	language=[x64]Assembler,
	basicstyle=\fontsize{7.5}{9}\ttfamily,
	keywordstyle=\bfseries\color{darkgray},
	breaklines=true,
	mathescape=true,
	keepspaces=true,
	showspaces=false,
	showstringspaces=false,    
	xleftmargin={0.55cm},
	firstnumber = 0,
	numbers=left,%
	numberstyle={\footnotesize \color{mGreen}},
	numbersep=5pt, 
}

\lstdefinestyle{birstyle}{
	language=[x64]Assembler,
	basicstyle=\fontsize{7.5}{9}\ttfamily,
	breaklines=true,
	mathescape=true,
	keepspaces=true,
	showspaces=false,
	showstringspaces=false,    
	xleftmargin={0.15cm},
	numbersep=2pt, 
	morekeywords={Address, Imm64, jmp, var, Const, Type_Imm, BLE_Exp, Exp_Den, assign, halt}, 
}

\lstdefinestyle{cstyle}{
	commentstyle=\color{mGreen},
	keywordstyle=\bfseries\color{mGreen},
	numberstyle=\tiny\color{mGray},
	stringstyle=\color{mPurple},
	breakatwhitespace=false,         
	breaklines=true,                 
	captionpos=b,                    
	keepspaces=true,                 
	showspaces=false,    
	numbers=none,          
	showstringspaces=false,
	showtabs=false,                  
	tabsize=2,
	xleftmargin={0.2cm},
	emph={int,char,double,float,unsigned, const, long},
	emphstyle={\bfseries\color{purple}},
	language=C
}

\lstdefinestyle{mlstyle}{
	language=caml,
	columns=[c]fixed,
	basicstyle=\small\ttfamily,
	keywordstyle=\bfseries,
	upquote=true,
	breaklines=true,
	showstringspaces=false,
	stringstyle=\color{blue},
	xleftmargin={0.1cm},
	aboveskip=1ex,
	literate={'"'}{\textquotesingle "\textquotesingle}3,
	morekeywords={new, out, event, assume}
}

\usepackage{marginnote}

\begin{document}

\begin{revision}
    
    \onecolumn
	\input{revision.tex}
    \twocolumn
\end{revision}

\title{Symbolic Parallel Composition for Multi-language Protocol Verification\\ {\normalsize Authors’ version; to appear in the 38th IEEE Computer Security Foundations Symposium}}

\author{
	\begin{tabular}[t]{c}
		\IEEEauthorblockN{Faezeh Nasrabadi} \\
		\IEEEauthorblockA{\textit{CISPA Helmholtz Center for}\\ 
			\textit{Information Security} \& \\ 
			\textit{Saarland University}\\
			faezeh.nasrabadi@cispa.de}
	\end{tabular}
	\hspace{0.5cm}
	\begin{tabular}[t]{c}
		\IEEEauthorblockN{Robert Künnemann}\\
		\IEEEauthorblockA{\textit{CISPA Helmholtz Center for}\\
			\textit{Information Security}\\
			robert.kuennemann@cispa.de}
	\end{tabular}
	\hspace{0.5cm}
	\begin{tabular}[t]{c}
		\IEEEauthorblockN{Hamed Nemati} \\
		\IEEEauthorblockA{\textit{Department of Computer Science}\\
			\textit{KTH Royal Institute of Technology}\\
			hnnemati@kth.se}
	\end{tabular}
}

\maketitle

\begin{abstract}
	The implementation of security protocols often combines different languages. This practice, however, poses a challenge to traditional verification techniques, which typically assume a single-language environment and, therefore, are insufficient to handle challenges presented by the interplay of different languages. 
	To address this issue, we establish principles for combining multiple programming languages operating on different atomic types using a symbolic execution semantics.
    This facilitates the (parallel) composition of labeled transition systems, improving the analysis of complex systems by streamlining communication between diverse programming languages.
    By treating the Dolev-Yao (DY) model as a symbolic abstraction, our approach eliminates the need for translation between different base types, such as bitstrings and DY terms.
	Our technique provides a foundation for securing interactions in multi-language environments, enhancing program verification and system analysis in complex, interconnected systems.
\end{abstract}

\begin{IEEEkeywords}
	distributed systems security, formal methods and verification, security protocols
\end{IEEEkeywords}

\section{Introduction}\label{sec:introduction}

In the rapidly evolving landscape of computer programming, it has become a norm that different programming languages coexist and interact within the same application. This is especially pronounced in complex systems like network protocols and operating systems, where components written in different languages must communicate seamlessly. Traditional approaches in program verification and system analysis often fall short in these multi-language environments, as they typically assume a homogeneous language framework. This assumption overlooks the challenges presented by the interplay of different programming languages, each with its unique syntax, semantics, and operational paradigms.

To overcome this limitation, we establish principles upon which two languages that operate on different atomic types can be combined. A typical use case is the analysis of network protocol
implementations. 
At a minimum, they combine a party written
in a real programming language, their communication partner(s)
operating by specification and modeled abstractly, e.g., in the applied pi calculus, and an attacker, which is underspecified but usually
limited by some threat model, e.g., the Dolev-Yao (DY) model or the
cryptographic model of a time-bounded probabilistic Turing machine.
As protocol properties extend over multiple parties in the presence of an
attacker, an implementation-level
analysis needs to reason about these types of components and their
interactions.

To this end, we extend \emph{parallel asynchronous composition},
which combines two systems communicating with an unspecified `outside' into a single interacting system.
The state of art~\cite{backes2016computational,sprengerIglooSoundlyLinking2020,arquint2022sound,arquint2023generic,sammlerDimSumDecentralizedApproach2023} describes the heterogeneous system as a composition of labeled transition systems (LTS).
LTS are very flexible; they can abstract any programming language.
Hence, the composition of LTS is the key to capturing cross-language
communication, be it at
runtime~\cite{sprengerIglooSoundlyLinking2020} or compile time \cite{sammlerDimSumDecentralizedApproach2023}.
In practice, protocols consist of components in different languages (e.g., the Apache server communicates with Firefox in TLS) and an altogether unknown attacker. 
Current composition approaches insist on \emph{translating} base values that are truly incompatible, e.g., bitstrings and abstract DY terms. This leads to shortcomings that we describe in detail (and solve) in \autoref{sec:parallel}.
In short, the translation approach is:
\begin{itemize}

    \item \textbf{Hard to apply due to strong parsing assumptions:} 
For instance, keys must always be syntactically distinguishable from
bitstrings used elsewhere and network messages must use known encodings~\cite{arquint2022sound,arquint2023generic}. We can avoid this assumption by not
requiring a `universal' translation a~priori, but instead by tracking what the
application actually does.
We elaborate on this in~\autoref{sec:parsing} and use~\Cref{ex:BS-manipulation} to show how we solve this problem.
    \item \textbf{Limited in the ability to capture adversarial bit-level reasoning:} 
        The translation approach notoriously struggles with mixed
        values, for instance, abstract encryption terms or keys that the
        implementation manipulates on the bit level. 
        In~\autoref{lose-bit-info}, we further explain this issue and discuss our solutions using~\Cref{ex:Transferable-equalities,ex:masked-encryption}.
    \item \textbf{Not truly versatile:} The complexity-theoretic
    computational attacker, e.g., is not compatible with
    standard language semantics.
    We argue the compatibility of our framework with a computational attacker in~\autoref{sec:comp-attacker}.
\end{itemize}

We solve these issues by forgoing the translation step between such different
base types as, e.g.,  bitstrings and DY terms.
The crux is, in our view, that the DY model is a symbolic
abstraction (it is sometimes called the `symbolic
model of cryptography'~\cite{abadi2000reconciling}),
whereas the translation approach and the above method of composition
treats DY terms as if they were concrete values (e.g., see~\cite[Sec.\ 4.2]{arquint2022sound}).
The first two of the above issues are artifacts of this mismatch.

Consequently, the DY model ought to be composed with an LTS that
describes interacting components at the same level of abstraction,
that is, with \emph{symbolic execution semantics}. 
Symbolic execution follows the program, assuming symbolic values
(i.e., variables at the object level, so neither program variables nor
meta-mathematical variables) for inputs and thus computing symbolic
expressions instead of concrete values. Symbolic values can form the
`glue' for communication, allowing us to describe message-passing
without translating values from one semantics into the other.

Assuming we have a symbolic execution semantics---devising them is a standard task---we can define a class of LTS with a little more structure, aka \textit{symbolic LTS}, and a new parallel
composition operator. The operator exploits symbols for
communication and covers the transfer of logical statements made about
symbols between both semantics.

This paper is organized into two major parts: \autoref{sec:parallel}
builds up our framework (source is available at~\cite{SymbolicParallelComposition})
and introduces the necessary background
whenever needed. It starts from LTS and traditional composition,
discusses the aforementioned problems in detail, then provides our new
method for composition, along with helpful results for composition,
refinement and DY attackers. It presents a \emph{framework} for
multi-language composition.
The second part (\autoref{sec:case-study}) presents a challenging application: we instantiate our framework with different languages for the representation of machine code,
for the specification of other parties and for the specification of
the threat model,
and demonstrate 
the sound extraction of a protocol model from its low-level implementation.
Our soundness result ensures the end-to-end correctness of our toolchain, distinguishing it from existing works~\cite{nasrabadi2023cryptobap,aizatulinVerifyingCryptographicSecurity2015}.
Finally, to show the efficacy of our framework, we apply it to verify the TinySSH and WireGuard protocols.
To summarize, we make the following contributions:

\begin{itemize}
    \item We propose a framework for the parallel asynchronous
        composition of components written in different languages 
        with applications for various methods of analysis, e.g.,
        secure compilation, code-level verification, model extraction (such as ours), or monitoring (see also \autoref{sec:relatedwork}).
    \item Using this framework and additional theorems, we
       support integrating DY attackers into arbitrary languages.
       This is necessary to make the end-to-end proof feasible.
    \item We discuss three methods of improving symbolic execution
        engines with DY support using combined deduction relations (i.e., $\compdedrelfun{}{}$) in \autoref{subsec:DedComb}.\footnote{We use \fstlangcol{RoyalBlue, \ math \ bold}, \sndlangcol{RedOrange, \ sans \ serif}, and \trdlangcol{Plum, \ typewriter} to differentiate between different languages~\cite{patrignani2020should}. Elements common to all languages, including symbols, are typeset in \textit{black italics}.}
\item We formalized our framework and proved its soundness and the DY attacker and library properties in HOL4~\cite{hol4}.
\item We extend \CryptoBap~\cite{nasrabadi2023cryptobap} toolchain and provide a sound, mechanized verification methodology for the verification of ARMv8 and RISC-V machine code. Thanks to our framework, we simplify the proofs in~\cite{nasrabadi2023cryptobap} and fully mechanize them, which was previously unrealistic due to the complexity of the employed computational soundness framework and lack of compositionality.
\item We took the opportunity to translate the symbolic results from \CryptoBap{} into \Sapic ~\cite{cheval2022sapic+}, a calculus that (soundly) translates to a range of protocol verification backends such as \Tamarin~\cite{meier2013tamarin}, \ProVerif~\cite{blanchet2001efficient} and \DeepSec~\cite{cheval2018deepsec}.
\item We compare the performance of \Sapic's backends by \revised{\infruleref{Rc-m3}{RC,M3}}{proving mutual authentication and forward secrecy within the symbolic model for the implementations of TinySSH and WireGuard.}
\end{itemize}

\section{Parallel Composition of Symbolic Semantics}
\label{sec:parallel}

\reversemarginpar
We now present our framework for the composition of symbolic labeled transition systems, starting with revisiting the standard definition of LTS and the conventional \revised{\infruleref{Rc-m1}{RC,M1}}{communicating sequential processes (CSP)-style~\cite{brookes1984theory} asynchronous parallel composition\footnote{
	We prefer a less descriptive, but shorter name, assuming that modern systems require both synchronous and asynchronous transitions.}. }
\normalmarginpar
We use a few illustrative examples to better highlight the translation approach's limitations, which we compensate for by a novel form of parallel composition in a symbolic semantics.
In our framework, we distinguish the specific roles of the DY attacker and the DY library. 
Also, we discuss its capability to deal with other attackers alongside the DY attacker. 
Finally, we demonstrate the correctness of our approach 
and, for each theorem, provide access to proofs that we mechanized in HOL4.

\subsection{LTS and their composition}\label{lts}
\revised{\infruleref{Rc-w4}{RC,W4}}{LTS provides a generic semantic model for capturing the operational semantics of systems~\cite{10.1016/j.entcs.2007.05.019,plotkin1980operational}.}
An LTS consists of a set of states (aka configurations) $\configSpace$
connected by a transition relation $\transrel{\eventsym}{}{}{} \subseteq \configSpace \times \eventSpace \times \configSpace $ 
that releases an event $\eventsym \in \eventSpace$ when the system moves
between states, and an initial state $\config \in \configSpace$
within that space. Given that a language has a formalized
semantics, a program behavior can typically be described
as an LTS. Thus, it is interesting to combine LTS to reason about
heterogeneous systems, wherein some transitions are asynchronous, e.g.,
 programs are performing internal computations independently, while
others are synchronized, e.g., one program sends a message and the other one
receives it.

CSP-style asynchronous parallel composition
supports both types of transitions and can be applied to LTS. 
Transitions are synchronous if both carry the same event ($\eventsym \in \eventSpace_1\cap \eventSpace_2$), and all others are asynchronous. Hence, in a composed state $(\config_1,\config_2)$,
we move synchronously to $(\config_1',\config_2')$ with event
$\eventsym$ provided \emph{both} systems can move 
($\transrel{\eventsym}{}{\config_1}{\config_1'}$
and
$\transrel{\eventsym}{}{\config_2}{\config_2'}$)
and otherwise ($\eventsym \notin \eventSpace_1\cap \eventSpace_2$),
we move to  
$(\config_1',\config_2)$
or
$(\config_1,\config_2')$
if either of the systems can make a transition.

Synchronizing events can be used to transmit messages~\cite{backes2016computational,sprengerIglooSoundlyLinking2020}.
For example, when combining two systems $A$ and $P$ with a shared event $\AtoP{ m }$, system $A$ can have a rule that determines $m$ from its current state, whereas $P$ has a rule that non-deterministically accepts $\AtoP{ m^* }$ for any $m^*$ and incorporates it into the follow-up state.
Combining both systems via asynchronous parallel composition, we
obtain synchronous message passing from $A$ to $P$. 

\subsection{Message passing and Dolev-Yao attackers}
\label{subsec:mess-DY}
A very appealing proposal is to let $A$ designate a DY attacker
and $P$ a program in some general-purpose language, as to obtain
a semantics to reason about the interaction of any such $P$ with
a network adversary. 
Most recently, this approach was used to leverage separation logic for
the verification of network systems~\cite{sprengerIglooSoundlyLinking2020,arquint2022sound,arquint2023generic},
and earlier to
provide sound \newrevised{analyses} for Dalvik bytecode~\cite{backes2016computational}.

The DY model is a model of cryptography where the
attacker only makes deductions defined by a set of rules. It
has been enormously successful in verifying security
protocols, as it automates the
verification procedure~\cite{barbosa2021sok}. 
These rules do not necessarily cover all possible attacks and require additional
justification~\cite{abadi2000reconciling,gupta2005towards}. Typically, the DY
attacker and the protocol share an unbounded set of \emph{names} that
represents keys and other hard-to-guess values. The model ensures
the attacker and protocol always draw fresh names, hence
key collisions are \newrevised{impossible}. Names and public values can be combined
with free function symbols to terms. E.g., $\senc(m,k)$ is a term that represents
an encryption.
It is not interpreted further.
This so-called \emph{term algebra} is complemented
by a small set of rules that allows
operations beyond the application of these symbols. E.g., a rule for
decryption that says from $\senc(m,k)$ and $k$, the attacker learns
$m$. We will make these notions explicit later.
When parallel composing a DY attacker with a language where keys and messages are represented as bitstrings, it is necessary to translate DY terms to bitstrings and vice versa. This, however, has several caveats. 

\subsubsection{Parsing assumptions}\label{sec:parsing}

First, it requires strong and unrealistic parsing
assumptions to transform bitstrings back into terms that have more
structure. For instance, keys must always be distinguished from
bitstrings used elsewhere~\cite{arquint2023generic}. When we consider
the space of AES keys, which (in reality) covers all bitstrings of length 128 (or
192 or 256), this requires (artificial) tagging to
distinguish those from other bitstrings of that size,
which real-world implementations do not have and
actively avoid for performance.
Another issue is the use of bitstring manipulation for message
formatting. 

\begin{example}[Bitstring manipulation]\label{ex:BS-manipulation}
	Concatenation is essential in the implementation of crypto protocols.
        It is associative and, hence, not easy to reason about automatically; thus, usually, this operation is not part of the DY term algebra. 
        Consider the example where
        a message $ \fstlangcol{m} $ 
        is concatenated with its length
        to simplify
        parsing.
        Without further workarounds,
        the DY attacker can not determine $ \fstlangcol{m} $ from $ \senc(\fstlangcol{m \|}\fstlangpred{len}\fstlangcol{(m)},k) $,
        even if it possesses the encryption key $ k $.
		As we show later, the DY attacker can derive $ \fstlangcol{m} $ from $\fstlangcol{m \|}\fstlangpred{len}\fstlangcol{(m)}$ by employing the deduction combinator $\dedrelcombbit$ in~\autoref{comb-ded-bit} defined in~\autoref{comb-reason}.
\end{example}

The translation approach supports message formatting, of course, otherwise it would be impractical. It works around this issue by modeling every message format that is used as a DY function symbol ~\cite[Sec.\ 3.1]{arquint2022sound}.
For full TLS, there are at least 189 message formats~\cite [Sec.\ 5.1]{ammann2024}.
Clever refactoring may reduce this number (formats can be nested), but this is non-trivial and tedious.
Most importantly, we would like our mechanism to be protocol-agnostic, even if it is tied to a particular set of cryptographic functions.
\reversemarginpar
\revised{\infruleref{Rc-w2}{RC,W2}}{In contrast, techniques like
$\mathit{DY}^*$ and Comparse \cite{dystareurosp2021,wallezComparseProvablySecure2023}
	integrate bit-level and DY reasoning within the same tool, \newrevised{enabling} the analysis of a (protocol-specific) set of message formats at the bit-level and then performing a DY analysis on abstract types. This avoids the
problem and is discussed in~\autoref{sec:relatedwork}.}
\normalmarginpar

\subsubsection{Loss of bit-level information}
\label{lose-bit-info}
Manipulating DY terms in the context of another language's semantics produces non-DY bitstrings that cannot be properly represented and translated back into their correct form. 
Therefore, these bitstrings become untraceable to their DY origins and an irreversible element to the transformation.
As a result, translation approaches weaken the DY attacker in reasoning about the messages altered by a protocol party using a different language.
E.g., say $A$ wants to learn $P$'s secret $s$ and
can trick $P$ into encrypting
$s \fstlangcol{+ 0x1}$ with a known key $k$.
The DY attacker receives a bitstring corresponding to
$\senc(s\fstlangcol{+ 0x1},k)$, and after decrypting
with $k$, 
has to recognize the transformation applied to $s\fstlangcol{+ 0x1}$
(and that it requires
subtracting $\fstlangcol{0x1}$).
Examining the huge number of possible
transformations is out of the question, particularly when considering
Turing-complete machine semantics (e.g., the wrappers in~\cite{sammlerDimSumDecentralizedApproach2023}).
Typically,
as bitstring addition $ \fstlangcol{+} $ does not correspond to the image of a term constructor,
such unknown bitstrings are translated into garbage DY
terms~\cite{sprengerIglooSoundlyLinking2020} or terms
we do not know~\cite{arquint2022sound,arquint2023generic}.
In contrast to message formats, this output was unintended.
Using~\Cref{ex:Transferable-equalities,ex:masked-encryption}, we explain our solutions to this problem in~\autoref{subsec:DedComb}.

\subsubsection{Not truly versatile, compatibility with computational model}\label{sec:comp-attacker}

The DY model is a symbolic abstraction that is well-accepted in
protocol verification, but not throughout information security.
It is useful to
be able to replace DY attackers with computational attackers for
flexibility or to validate the DY attacker's soundness.
This is incompatible or difficult, depending on how the translation approach is realized.
In~\cite{arquint2022sound,arquint2023generic},
a function translates from 
terms to bitstrings (i.e., the inverse direction to parsing discussed above).
In the computational model, this relationship is not
functional. For instance, a DY term representing a key, i.e.,
a \emph{name}, may translate to many different bitstrings, depending on
how they are sampled. 
Consequently, the computational attacker in these works is not
an attacker in the traditional sense (an arbitrary
probabilistic algorithm limited only in runtime) but the DY attacker
inside a function translating from and to bitstrings.

Fortunately, a long line of work on computational soundness~\cite{abadi2000reconciling} explored requirements for such a translation, which must be probabilistic. Alas, known results come with a long list of requirements, both on programs and cryptographic primitives they use, that are hard to fulfill. 
To even formulate these requirements, the target semantics
need to be equipped with a probabilism, non-determinism for
communication and a notion of polynomial runtime in the length of some
parameter that governs the key size and similar parameters.
While there are methods to encode all of these, programming languages
are rarely formalized with these features in mind.
We can point to Aizatulin's Ph.D. thesis \cite{aizatulinVerifyingCryptographicSecurity2015} as a case study for such a semantics and the required technical machinery.

\subsection{Symbolic Execution Semantics}
\label{sec:symb-exec-sem}
Symbolic execution explores all program execution paths using symbolic values---introduced at the object level---instead of concrete ones for inputs.
An example is a language with a memory that maps registers to
bitstrings. Its symbolic execution allows the memory to map
registers to either symbols or bitstrings.
Starting from an initial symbolic state, the execution explores all possible paths and collects the execution effects in a final symbolic state for each path.
Each symbolic state, in addition to a map from variables to symbolic expressions (i.e., where symbols represent initial state variables), also contains a path condition that is a logical predicate describing what is known about the symbol.
For instance, $r_A = \mathtt{0x0} \lor r_A = \mathtt{0x1}$ if register $r_A$ is known to be either $0$ or $1$ because it passed some condition. 
To combat the path explosion problem, symbolic execution
engines make logical deductions on these predicates to prune paths
that are unreachable \newrevised{(e.g., using an SMT solver)}. The more powerful the deduction engine, the fewer
paths need to be explored, but the more computationally expensive
these deductions are.

We capture these elements---symbols, predicates, and deductions––-by giving our LTS more structure. Let $\silentevent$ be the silent transition, then:

\begin{definition}[Symbolic LTS]\label{def:symbolic-lts}
    A symbolic LTS is 
    an LTS $(\stateSpace,\eventSpace,\transrel{}{}{}{})$
    for which there is 
    a symbol space $\symbSpace$,
    a predicate space $\predSpace$,
    and a deduction relation $ \dedrel \ \subseteq 2^{ \predSpace } \times\predSpace $ 
    such that:
    \begin{itemize}
        \item $\stateSpace = 2^{ \symbSpace } \times 2^{ \predSpace } \times \configSpace$ for some state space $\configSpace$ and
        \item \revised{\infruleref{Rc-m2}{RC,M2}}{For any
        predicate set $ \predset $,
        predicate $ \pred $,
        symbols set $ \symbset $, and
        state $ \config $,
        we have:} $\predset \dedrel \pred \implies
            \transrel{\silentevent}{}{(\symbset,\predset,\config)}{(\symbset,\predset\cup\set{\pred},\config)}$.
    \end{itemize}
    For brevity, we denote such LTS with
    $(\symbSpace,\configSpace,\eventSpace,\transrel{}{}{}{},\predSpace,\dedrel)$.
\end{definition}

The second condition establishes the relation between \newrevised{the logical deduction
relation
which is language-specific},
and the current predicate set: logical deductions can be made
at any time, and the knowledge we conclude (encoded inside the predicate) is added to the symbolic state. 
Typically, the state space $\configSpace$ and event space
$\eventSpace$ are built on the symbol space $\symbSpace$, e.g., in the example above, the
symbolic memory was a function from registers to
the union of bitstrings and the set of symbols $\symbset \subseteq \symbSpace$. 
This is only implicit in the
mathematical notation, it is, however, explicit in our HOL4 formalization, where the
types of $\configSpace$ and $\eventSpace$ are parametric in the
(polymorphic) type $\symbSpace$.
The first element $\symbset$ mainly tracks which symbols have been used so far, increasing monotonically.

\revised{\infruleref{Ra-w1.a}{RA,W1.a}}{Every symbolic LTS, also referred to as a component, 
must transmit only references to their messages in the form of symbols to other components.} 
\revised{\infruleref{Ra-w1.d}{RA,W1.d}}{Symbols relate to the values that are transmitted like a variable $n$ relates to the set of integers, i.e., as a representation. When a value is manipulated, the relation between the original and the changed value, each represented by a different symbol, is itself represented with a predicate connecting the two symbols. 
Consequently, a symbol always signifies the same value (in a run), and the predicates associated with distinct components articulate the same properties.}

\subsection{Symbolic Parallel Composition}\label{subsec:symbolic-parallel-composition}

We now define a parallel composition that behaves like
CSP-style asynchronous parallel composition
but has an important twist:
it is parametric in a \emph{combined deduction relation}, which serves to transfer judgments from one system into the
other. In the follow-up, we show that there are several ways to
define this that increases the set of possible deductions and, thus, the precision of the analysis, while also being compatible with almost all judgments made in programming
languages.\footnote{
From hereon, we will combine different systems with oftentimes
incompatible base types. To make it easier for the reader to type-check
our statements, we will use colors to remark which system we speak of.
}
Let $\proj{i}{} \ : 2^{\subone{\predSpace} \uplus \subtwo{\predSpace}} \to 2^{\predSpace_{i}} $ denote the projection\footnote{\newtext{We present this using a disjoint union ($\uplus$) for familiarity and simpler presentation, while we employ a sum type in our HOL4 formalization.}}
to $i\in\set{\colone,\coltwo}$, then:

\begin{definition}[Symbolic Parallel Composition]\label{def:symbolic-parallel-composition} 
    Given two symbolic LTS
    $S_i = (\symbSpace,\configSpace_i,\eventSpace_i,\transrel{}{i}{}{},\predSpace_i,\dedrel_i)$, 
    $i\in\set{\colone,\coltwo}$
    with \emph{identical symbol space} $\symbSpace$
    and a combined deduction relation
    $\dedrelcomb \subseteq 2^{(\subone{\predSpace} \uplus \subtwo{\predSpace})} \times (\subone{\predSpace} \uplus \subtwo{\predSpace}) $,
    we define their symbolic parallel composition
    $\Parallel[\dedrelcomb]{\SysFst[ _1 ]}{\SysSnd[ _2 ]}$ as the symbolic LTS
    $
    (\symbSpace, 
    \subone{\configSpace}\times \subtwo{\configSpace},
    \subone{\eventSpace} \cup \subtwo{\eventSpace}, 
    \transrel{}{\colonetwo}{}{},
    \subone{\predSpace} \uplus \subtwo{\predSpace}, \dedrelcomb)$,
    where 
    \begin{itemize}
        \item $\transrel{}{\colonetwo}{}{}$ moves asynchronously, i.e., 
        either $\transrel{\subone\eventsym}{\colonetwo}{\tuple{\symbset,\comppredset,\subone{\config},\subtwo{\config}}}{\tuple{\symbset',\comppredset',\subone{\config'},\subtwo{\config}}}$ 
        or $\transrel{\subtwo\eventsym}{\colonetwo}{\tuple{\symbset,\comppredset,\subone{\config},\subtwo{\config}}}{\tuple{\symbset',\comppredset',\subone{\config},\subtwo{\config'}}}$, 
            if,
            for $i \in\set{\colone,\coltwo}$,
            we can move with 
            $ \eventsym_i \in \eventSpace_i \setminus (\subone{\eventSpace} \cap \subtwo{\eventSpace})$, i.e.,
            $ \transrel{\eventsym_i}{i}{\tuple{\symbset,(\proj{i}{\comppredset}),\config_{i}}}{\tuple{\symbset',(\proj{i}{\comppredset'}),\config'_{i}}}$,
            keeping the 
            complement's \footnote{i.e.,  $ \overline i \in\set{\colone,\coltwo}$ with $\overline i \neq i$.}
            predicate set untouched $ \proj{\overline i}{\comppredset} = \proj{\overline i}{\comppredset'} $, or
        \item $\transrel{}{\colonetwo}{}{}$ moves synchronously, i.e. $\transrel{\eventsym}{\colonetwo}{\tuple{\symbset,\comppredset, \subone{\config},\!\subtwo{\config}}}{\tuple{\symbset'\!,\comppredset',\subone{\config'},\subtwo{\config'}}}$, if, for $i\in\set{\colone,\coltwo}$, $ \transrel{\eventsym}{i}{\tuple{\symbset,(\proj{i}{\comppredset}),\config_{i}}}{\tuple{\symbset'_{i},(\proj{i}{\comppredset'}),\config'_{i}}} $, $ \eventsym \in \subone{\eventSpace} \cap \subtwo{\eventSpace} $, and $ \symbset' = \subone{\symbset'} \cup \subtwo{\symbset'} $.
        \item \revised{\infruleref{Ra-w3.a3}{RA,W3.a3}}{From the second condition of~\autoref{def:symbolic-lts} we have: $ \comppredset \dedrelcomb \predcomb \implies \transrel{\silentevent}{\colonetwo}{\tuple{\symbset,\comppredset, \subone{\config},\subtwo{\config}}}{\tuple{\symbset, \comppredset\cup\set{\predcomb},\subone{\config},\subtwo{\config}}}$.}
    \end{itemize}
\end{definition}

\revised{\infruleref{Ra-w3.a2}{RA,W3.a2}}{
~\autoref{def:symbolic-parallel-composition} preserves fundamental properties of parallel composition like symmetry and associativity 
(see}~\href{https://github.com/FMSecure/CryptoBAP/tree/main/HolBA/src/tools/parallelcomposition/generaldeduction/derived_rules_generaldeductionScript.sml#L121}{\itshape\underline{Symmetry}} \revisedtxt{and}~\href{https://github.com/FMSecure/CryptoBAP/tree/main/HolBA/src/tools/parallelcomposition/generaldeduction/derived_rules_generaldeductionScript.sml#L285}{\itshape\underline{Associativity}} \revisedtxt{for the mechanized proofs in HOL4).} 

Note that even if $\dedrelcomb$ is empty
(short:
$\Parallel[]{}{}
\defeq \Parallel[\emptyset]{}{}
$)
the symbolic parallel composition is different from the classical parallel composition of the corresponding LTS. The symbol set is shared between both symbolic LTS even when they move asynchronously.
Typically, symbolic LTS uses the symbol set to ensure that new symbols are fresh; as we use symbols for communication, we want to ensure they are globally fresh.

Moreover, if $\dedrelcomb$ is not empty, it allows deriving judgments
in one system from judgments in the other system. Of course, we want
to avoid this relation to be overly tied to one or the other system.
Before we discuss how to do that, we will showcase how the DY model is represented as a symbolic LTS. This provides us
with concrete examples to illustrate how $\dedrelcomb$ can overcome
the issues from \autoref{subsec:mess-DY} (cf.~\Cref{ex:masked-encryption,ex:Transferable-equalities} for their solutions). 

\subsection{DY Attackers}
\label{sec:DYattackers}
The DY model considers an attacker that exploits logical
weaknesses in a protocol, but is not able to break cryptographic
primitives. Cryptography is assumed to be perfect; events that can
occur in the real world, albeit with negligible probability---for
example, guessing a key---are altogether impossible in this model.
A small set of rules governs how messages can be manipulated on an
abstract level; every other manipulation is excluded.
Concretely, messages are modeled as \emph{terms}\footnote{\newtext{The DY attacker and the DY library (but not the program) use the same terms. To simplify the presentation, we typeset them in \textit{black italics}, as otherwise, they would be orange \emph{and} purple.}}. 
\newrevised{
Constants are taken from an infinite set of
names $\allnames$, divided into public names $\pubnames$ (e.g., agent names)
and secret names $\privnames$ (like keys and nonces).}
We also assume a set of variables $\Vars$ for values that the DY
attacker receives. 
The set of terms $\Terms$ is then constructed over names in
$\allnames$, variables in $\Vars$ and applications of function
symbols in $\Funcs$ on terms. Let $f\in \Funcs^n$ denote a function
symbol with arity $n$.
For the moment, we consider only two function symbols, 
$\Funcs = \set{ \senc, \sdec } = \Funcs^{\newrevised{2}}$. 
The term $\senc(m,k)$ models the symmetric encryption of
another term $m$ with the key $k \in \privnames$.
A set of equations $E \subset \Terms \times \Terms$ provides these
terms with a meaning. Let us define
$ E = \set{  \sdec(\senc(x,y),y) = x} $
to account for the fact that decryption reverses encryption (for the
same key). We can define an equivalence relation $=_E$ as the smallest
equivalence relation containing 
$E$ that is closed under the application of function symbols and substitution of variables by
terms. Now, 
$ \sdec(\senc(m,k),k) =_E  m$. 

\begin{figure}
\centering
\adjustbox{varwidth=\linewidth,scale=0.95}{%

\[
\begin{prooftree}	
\hypo{
\knowledge{t} \in \knowledgeset 
}
\infer1[\rKZero]{ \dydedrelfun{\knowledgeset}{\knowledge{t}} }
\end{prooftree}
\quad
\begin{prooftree}	
    \hypo{ \dydedrelfun{\knowledgeset}{\knowledge{t_1}} \cdots \dydedrelfun{\knowledgeset}{\knowledge{t_n}} \ \ f \in \Funcs^n }
\infer1[\rAppl]{
        \dydedrelfun{\knowledgeset}{\knowledge{f(t_1,\ldots,t_n)}} }
\end{prooftree}
\]
\[
\begin{prooftree}	
	\hypo{
		n \in \pubnames
	}
	\infer1[\rPubName]{ \dydedrelfun{\knowledgeset}{\knowledge{n}} }
\end{prooftree}
\quad
\begin{prooftree}	
	\hypo{ \dydedrelfun{\knowledgeset}{\knowledge{t_1}}}
	\hypo{ \dydedrelfun{\knowledgeset}{\dyEqual{t_1}{t_2}} }
	\infer2[\rSubst]{ \dydedrelfun{\knowledgeset}{\knowledge{t_2}}  }
\end{prooftree}
\]
\[
\begin{prooftree}	

    \hypo{ t_1 =_E t_2}
    \infer1[\rEq]{
	\dydedrelfun{\knowledgeset}{\dyEqual{t_1}{t_2}}
    }
\end{prooftree}
\quad
\begin{prooftree}	
	\hypo{ \dydedrelfun{\knowledgeset}{\knowledge{x}}}
	\hypo{ \dydedrelfun{\knowledgeset}{\dyAliEq{x}{t}} }
	\infer2[\rAlSubst]{ \dydedrelfun{\knowledgeset}{\knowledge{t}}  }
\end{prooftree}
\]
}
\caption{The DY attacker's deduction rules.
Top-left to bottom-right: the attacker knows the messages it received and
can apply function symbols. 
If a name is public, the attacker knows this name, and equivalent names to fresh names are also fresh.
Equivalence modulo $E$ translates into an
equivalence judgment, and if a term is known, any equivalent terms are also known. 
Any terms that correspond to a given symbol are known if the symbol itself is known.
}
\label{fig:DYdeductionrules}

\end{figure}

\begin{figure*}
\centering
\adjustbox{varwidth=\linewidth,scale=0.9}{%
\[
\begin{prooftree}	
	\hypo{ x \notin \symbset }
	\hypo{\knowledgeset' = \knowledgeset \cup \{ \knowledge{x} \}}
	\infer2[\rPA]{	\transrel{\PtoA{x}}{_{\dycol{A}}}{\tuple{\symbset,\knowledgeset,\dystate}}{\tuple{\symbset \cup \{ x \}, \knowledgeset',\dystate}} }
\end{prooftree}
\quad
\begin{prooftree}	
	\hypo{ \knowledge{x} \in \knowledgeset }
	\hypo{ x \in \symbset }
	\infer2[\rAP]{	\transrel{\AtoP{x}}{_{\dycol{A}}}{\tuple{\symbset,\knowledgeset,\dystate}}{\tuple{\symbset,\knowledgeset,\dystate}} }
\end{prooftree}
\quad
\begin{prooftree}	
		\hypo{\begin{matrix}
		x \notin \symbset \\
		f \in \Funcs^n
	\end{matrix}}
	\hypo{\begin{matrix}
		\forall i \leq n: t_i \in \Terms  \hspace{1cm} 0 < n\\
		\knowledgeset' = \knowledgeset \cup \set{ \dyAliEq{x}{f(t_1,\ldots,t_n)} }
\end{matrix}}
        \infer2[\newrevised{\rAlias}]{\transrel{\Alias{x}{f(t_1,\ldots,t_n)}}{_{\dycol{A}}}{\tuple{\symbset,\knowledgeset,\dystate}}{\tuple{\symbset \cup \{x\},      \knowledgeset' ,\dystate}} }
\end{prooftree}
\quad
\]
\[
\begin{prooftree}	
	\hypo{\begin{matrix}
			n \in \privnames \\
			\freshness{n} \notin \knowledgeset
	\end{matrix}}
	\hypo{\begin{matrix}
			\knowledgeset' = \knowledgeset \cup \{ \freshness{n} \}
	\end{matrix}}
	\infer2[\rFreshSync]{	\transrel{\Syncfresh{n}}{_{\dycol{A}}}{\tuple{\symbset,\knowledgeset,\dystate}}{\tuple{\symbset,\knowledgeset',\dystate}} }
\end{prooftree}
\quad
\begin{prooftree}	
	\hypo{\begin{matrix}
			n \in \privnames \\
			\freshness{n} \notin \knowledgeset
	\end{matrix}}
	\hypo{\begin{matrix}
			\knowledgeset' = \knowledgeset \cup \{ \freshness{n},\knowledge{n} \}
	\end{matrix}}
	\infer2[\rFresh]{\transrel{\Silent{n}}{_{\dycol{A}}}{\tuple{\symbset,\knowledgeset,\dystate}}{\tuple{\symbset,\knowledgeset',\dystate}} }
\end{prooftree}
\quad
\begin{prooftree}
        \hypo{ \dydedrelfun{\knowledgeset}{\pi}}
        \hypo{\knowledgeset' = \knowledgeset \cup \set{ \pi} }
        \infer2[\rDed]{	\transrel{\sndlangcol\silentevent}{_{\dycol{A}}}{\tuple{\symbset,\knowledgeset,\dystate}}{\tuple{\symbset,\knowledgeset',\dystate}} }
\end{prooftree}
\] 
}
\caption{The transition relation rules for Dolev-Yao attacker model.}
\label{fig:DYtransitionrules}
\end{figure*}

The predicate set of \newrevised{the} DY attacker has three types of facts:
\newtext{their knowledge of a term $t$, written as $ \knowledge{t} $, is derivable from the set of predicates $\knowledgeset$ seen or derived so far,}
two terms are considered equivalent $\dyEqual{t_1}{t_2}$
according to $=_E$ and a name $n$ is fresh $\freshness{n}$. The
deduction relation (\autoref{fig:DYdeductionrules}, see caption)
 mostly
describes how $\knowledge{\cdot}$ is derived. 
$\dyEqual{t_1}{t_2}$ \newrevised{and $ \dyAliEq{x}{t} $ represent $=_E$ and $ \mapsto $ (i.e., denotes the mapping of variables to terms) at the logical level, respectively.}
With the deduction relation in place, we can now define the transition
relation (\autoref{fig:DYtransitionrules}). Besides \newrevised{the symbol set $\symbset$ and} the predicate set
$\knowledgeset$, the DY attacker is stateless, indicated by $\dystate$
for the empty state.

\newrevised{A message is received by synchronization with the event $\PtoA{\com{x}}$ emitted by another component.}
As the DY
adversary cannot process the incoming message type (e.g., bitstrings)
directly, we must assume $x$ is a symbol. Therefore, we set $\Vars = \symbSpace$.
Hence, in \infruleref{fig:DYtransitionrules}{\rPA},  $x$ is determined by the environment (e.g., the sending
component) and the attacker record the fact that it is known.

\newtext{
    If the predicate set $\knowledgeset$ witnesses that the symbol $x$ from the set $\symbset$ represents
    a value known to the DY attacker, the attacker can send $x$ to another component (\infruleref{fig:DYtransitionrules}{\rAP}). 
But not all knowledge predicates within $\knowledgeset$ are over symbols}; encryption
terms, for instance. Hence, \newrevised{the} \infruleref{fig:DYtransitionrules}{\rAlias} rule can be used to introduce a new
symbol, which can be transmitted. 
Recall that the \revised{\infruleref{Ra-w3.c}{RA,W3.c}}{\reversemarginpar\infruleref{fig:DYdeductionrules}{\rAlSubst} rule in~\autoref{fig:DYdeductionrules}} introduces the required
$\knowledge{\cdot}$-predicate via the deduction relation.
The transition rule \infruleref{fig:DYtransitionrules}{\rDed} integrates the deduction relation. 
It is simply the minimal rule required to satisfy the
condition in~\autoref{def:symbolic-lts}.
\normalmarginpar

Fresh names can be drawn by the attacker, but also by other
components (see \infruleref{fig:DYtransitionrulesLib}{\rFreshSync} in~\autoref{fig:DYtransitionrulesLib}).
\infruleref{fig:DYtransitionrules}{\rFresh} \newtext{is a synchronous step between the DY attacker and other
	components that} deals with the first case, where the attacker
learns the name, which is marked as fresh and thus `taken'. 
In the second case, another component, typically the crypto
library of some party, 
picks a key (or another high-entropy value) 
that is marked as fresh. 
The DY attacker must also mark those names as `taken', hence the other component synchronizes their picking of this value using the \newtext{synchronous \infruleref{fig:DYtransitionrulesLib}{\rFreshSync} rule in~\autoref{fig:DYtransitionrulesLib}.}
This synchronization is necessary, as \Cref{ex:DY-communicates-Lib} shows.

\begin{example}[DY communicates with crypto library]\label{ex:DY-communicates-Lib}
    To generate a random number, a request needs to be sent to the library (e.g., $\mathlabel{rng}$). 
	The library maintains a record of the generated random numbers within its predicate set (i.e., $ \{\libfreshness{n}\} $) to ensure the creation of unique names.
	The DY attacker has the ability to choose a name (e.g., $ n' $) as long as it differs from the choice made by the library for the program (i.e., $ n $).
 \begin{figure}[H]
    \centering
    \includegraphics[width=0.8\linewidth]{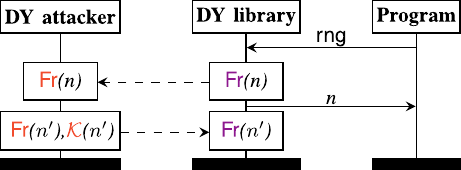} 
\end{figure} 
	The library's predicate set is updated using the synchronous \infruleref{fig:DYtransitionrulesLib}{\rFreshSync} rule in~\autoref{fig:DYtransitionrulesLib} and $ \freshness{n}$ is added to the predicate set of the attacker by the synchronous \infruleref{fig:DYtransitionrules}{\rFreshSync} rule in~\autoref{fig:DYtransitionrules}
	for the library's initial random number generation.
	The second update is performed by the attacker using the synchronous \infruleref{fig:DYtransitionrules}{\rFresh} rule in~\autoref{fig:DYtransitionrules} and the attacker is not able to pick the library chosen name $ n $ as $ \freshness{n}$ exists in the attacker predicate set $\knowledgeset$. Therefore, the attacker chooses a fresh name $ n' $, and the $ \libfreshness{n'}$ is added into the predicate set of the library by the synchronous \infruleref{fig:DYtransitionrulesLib}{\rFresh} rule in~\autoref{fig:DYtransitionrulesLib}.

\end{example}
In our examples,
\textbf{solid} arrows represent direct communications between components, while \textbf{dashed} arrows denote the implicit flow of facts between the DY library and the DY attacker. Also, each step within the action box of components signifies the logical predicates added to their predicate sets during execution.

Observe that DY attackers do not pick honestly generated names (e.g., in~\Cref{ex:DY-communicates-Lib})
due to the synchronization on freshness facts.
The \infruleref{fig:DYtransitionrules}{\rAlias} rule (\autoref{fig:DYtransitionrules}) only generates new symbols, not names \newrevised{(i.e., chooses $ x $ from the symbol set $ \symbset $)}.
Contrast this with \cite{sprengerIglooSoundlyLinking2020}, where the
syntactic structure of names binds them to protocol roles,
including the attacker,
or 
the follow-up work~\cite{arquint2022sound}  
where names do not need to carry structure, 
but a global restriction on traces 
is applied to ensure uniqueness.
In both cases, this aspect of the DY attacker is thus hard-coded
into the (global) trace model, which we can avoid.

\subsection{Dolev-Yao Libraries}\label{subsec:DYLib}
\begin{figure*}
\centering
\adjustbox{varwidth=\linewidth,scale=0.9}{%

\[
\begin{prooftree}	
	\hypo{n \in \privnames}
	\hypo{\begin{matrix}
			\libfreshness{n} \notin \libknowledgeset
	\end{matrix}}
	\hypo{\begin{matrix}
			\libknowledgeset' = \libknowledgeset \cup \{ \libfreshness{n} \}
	\end{matrix}}
	\infer3[\rFreshSync]{	\transrel{\Syncfresh{n}}{_{\libcol{L}}}{\tuple{\symbset,\libknowledgeset,\libstate}}{\tuple{\symbset,\libknowledgeset',\libstate}} }
\end{prooftree}
\quad
\begin{prooftree}	
	\hypo{n \in \privnames}
	\hypo{\begin{matrix}
			\libfreshness{n} \notin \libknowledgeset
	\end{matrix}}
	\hypo{\begin{matrix}
			\libknowledgeset' = \libknowledgeset \cup \{ \libfreshness{n} \}
	\end{matrix}}
	\infer3[\rFresh]{	\transrel{\Silent{n}}{_{\libcol{L}}}{\tuple{\symbset,\libknowledgeset,\libstate}}{\tuple{\symbset,\libknowledgeset',\libstate}} }
\end{prooftree}
\quad
\] 
\[
\begin{prooftree}	
	\hypo{\begin{matrix}
	y \notin \symbset 
	\end{matrix}}
	\hypo{\begin{matrix}
	\forall i \leq n: x_i \in \symbset 
	\end{matrix}}
	\hypo{\begin{matrix}
			f \in \Funcs^n
	\end{matrix}}
	\hypo{\begin{matrix}
		\libknowledgeset' = \libknowledgeset \cup \{ \callEq{y}{f(x_1,\ldots,x_n)}  \}
\end{matrix}}
	\infer4[\rFCall]{	\transrel{\fcall{f,x_1,\ldots,x_n,y}}{_{\libcol{L}}}{\tuple{\symbset,\libknowledgeset,\libstate}}{\tuple{\symbset \cup \{y\},\libknowledgeset',\libstate}} }
\end{prooftree}
\]
}
\caption{The transition relation rules for Dolev-Yao library model.}
\label{fig:DYtransitionrulesLib}

\end{figure*}

To equip a programming language with a DY semantics, we also need to
mark crypto operations as such, i.e.,
specify when a function output ought to be abstracted by an encryption
term like $\senc(\cdot,\cdot)$.
One way is to integrate the term algebra into the predicate space
$\predSpace$ and mark crypto outputs via equalities, e.g.,
a logical predicate saying `symbol $z$ is equivalent to the DY term
$\senc(x,y)$' (where $ x $ and $ y $ can be other symbols).
A more generic way to achieve the same effect is by composing the
function calls with a DY library that performs those abstraction steps.
\autoref{fig:DYtransitionrulesLib} shows how the composition can be done, with
\newrevised{the} \infruleref{fig:DYtransitionrulesLib}{\rFCall} applying a function symbol similar to \newrevised{the} \infruleref{fig:DYdeductionrules}{\rAppl} but including \newrevised{the} \infruleref{fig:DYtransitionrules}{\rAlias}.
\newtext{Like the DY attacker, the DY library is stateless.}

\revised{\infruleref{Ra-w3.d2}{{\small RA,W3.d2}}}{
	The deduction relation of the DY library is defined via an
        equivalence relation $ =_E $.
        In verification tools like \Sapic, 
        \newrevised{the relation $ =_E $ arises as the smallest congruence relation that is closed under substitutions and contains the set of equations $ E $ provided by the user.}
        For the sake of the formalization, $=_E$ is an arbitrary
        equivalence relation.
        The equations $E$ used in our case studies are provided in the \Sapic input file.
}

\begin{table}[]
	\begin{tabular}{@{}lll@{}}
		\toprule
		Event            & Purpose                        & Involved components         \\ \midrule
		$ \fcallsym $            & Library calls                  & Program and Library                         \\
		$ \Syncfreshsym $              & Calls to $\mathlabel{RNG}$                   & Program, Library and Attacker                       \\
		$ \AtoPsym , \PtoAsym $          & Network communication          & Program and Attacker                         \\\midrule
		$ \Silentsym $           & Ensure freshness               & Library and Attacker                        \\ \bottomrule
	\end{tabular}
	\caption{
		Summary of synchronizing events.
		$ \fcallsym $, $ \Syncfreshsym $, $ \AtoPsym$,  and $\PtoAsym $ synchronize with the program component,
		whereas $ \Silentsym $ is internal to the DY Libary and DY Attacker.}
	\label{tab:events}
	\end{table}

As we know the DY library and DY attacker and their respective
predicate sets, we can use the combined deduction relation 
\revised{\infruleref{Ra-w3.c}{RA,W3.c}}{$ \dedrelAvsL $ to
share the mapping predicates, as follows:}
\begin{align*}
	\libknowledgeset \uplus \knowledgeset \dedrelAvsL \ \dyAliEq{x}{t} \ \iff \ 
	\callEq{x}{t} \in \libknowledgeset \\
	\libknowledgeset \uplus \knowledgeset \dedrelAvsL  \ \callEq{x}{t} \ \iff \ 
	\dyAliEq{x}{t} \in \knowledgeset 
		\label{comb-ded-lib-dy}\tag{$\dedrelAvsL$}
\end{align*}
\revised{\infruleref{Ra-w1.b}{RA,W1.b}}{
\autoref{tab:events} summarizes the interface to the DY attacker and
library from the perspective of a protocol component, which could be,
for instance, a \birsymb{} program, as in our case studies.
For instance, if the protocol wanted to generate a random number, it
would use $\Syncfreshsym$, which synchronizes with 
\rFreshSync{} in 
\autoref{fig:DYtransitionrules}
and
\autoref{fig:DYtransitionrulesLib}.}

\begin{example}[Logical truth]\label{ex:Logical-truth}
	\newrevised{Predicates can be shared between components without explicit communication.}
 \begin{figure}[H]
    \centering
    \includegraphics[width=0.82\linewidth]{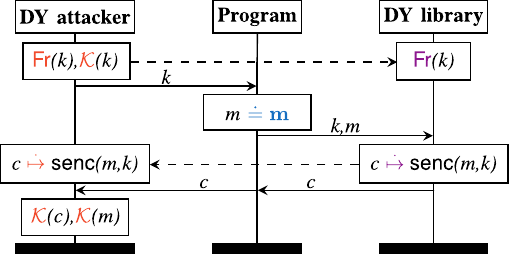} 
\end{figure} 
	Thus, the DY attacker can uncover the message $\fstlangcol{m}$ using the known key $ k $ and \revised{\infruleref{Ra-w3.c}{RA,W3.c}}{the mapping $\dyAliEq{c}{\senc(m,k)}$}, without communicating with the library.
	The steps to acquire the message $  \fstlangcol{m}  $ are as follows: 
	upon receiving $ c  $ from the program and obtaining \revised{\infruleref{Ra-w3.c}{RA,W3.c}}{$ \dyAliEq{c}{\senc(m,k)}$ through $ \dedrelAvsL $}, 
	the attacker uses \revised{\infruleref{Ra-w3.c}{RA,W3.c}}{the \infruleref{fig:DYdeductionrules}{\rAlSubst} rule} to get $ \knowledge{\senc(m,k)}$.
	Next, the attacker utilizes their knowledge and $ \sdec \in \Funcs^{\newrevised{2}} $ to learn $ \sdec(\senc(m,k),k) $ using the \infruleref{fig:DYdeductionrules}{\rAppl} rule in~\autoref{fig:DYdeductionrules}.
	Leveraging the relation $=_E$ detailed in~\autoref{sec:DYattackers}, along with the \infruleref{fig:DYdeductionrules}{\rEq} and \infruleref{fig:DYdeductionrules}{\rSubst} rules (\autoref{fig:DYdeductionrules}), the attacker obtains the knowledge of $ m $. 
\newtext{Without the mapping predicate linking the ciphertext and encryption term, the attacker would lack the necessary knowledge to apply the \infruleref{fig:DYdeductionrules}{\rAppl} rule for decryption,
	leaving the encryption term undisclosed.}
\end{example}

\Cref{ex:Logical-truth} shows how the DY library, the attacker, and the program cooperate when the library generates a ciphertext using an adversarial key.
As a nice extra, such a library allows us to
prove a composition property that is convenient when different programs use multiple libraries (cf.~\appendixorfull{sec:dy-lib-comp}).

\subsection{Deduction combiners}
\label{subsec:DedComb}

Symbolic parallel composition's strength lies in its ability to transfer judgments between systems. There is a trade-off between precision and generality. We discuss some
useful combiners from the most general to the most precise.

\subsubsection{Generic over- and under approximation}

In general, bitstring operations can reveal cryptographic information.
\Cref{ex:masked-encryption} shows how to under-approximate or over-approximate the adversaries’ capabilities on operating with bitstrings.

\begin{example}[Masked encryption key]\label{ex:masked-encryption}
	\revised{\infruleref{Ra-w1.b}{RA,W1.b}}{
    In this example, the attacker obtains a message $\fstlangcol{m}$ encrypted
    with a fresh key $k$, followed by the key masked with a known
    constant
    $\fstlangcol{ \mathtt{0xdeadbeef}}$.
Using the combined deduction relation 
\revised{\infruleref{Ra-w3.c}{RA,W3.c}}{$\dedrelAvsL$ (which is defined specifically for DY attacker and DY library),
the mapping $\dyAliEq{c}{\senc(m,k)}$} transfers from DY library to DY attacker.}

	\revised{\infruleref{Ra-w1.b}{RA,W1.b}}{
The last message $b$ ought to reveal the plaintext $\fstlangcol{m}$. In
the \newrevised{following}, we introduce 
an
\emph{over-approximating} deduction combiner $ \dedrelcomb^\top $ (\autoref{comb-ded-over-approx})
that allows 
the DY attacker to infer $ \knowledge{k} $ from $ \knowledge{b} $ and $ \progEqual{b}{k \fstlangcol{ \oplus \mathtt{0xdeadbeef}}}$
and thus
the plaintext $\fstlangcol{m}$ (from $ \knowledge{c} $, \revised{\infruleref{Ra-w3.c}{RA,W3.c}}{$ \dyAliEq{c}{\senc(m,k)} $} and $\knowledge{k} $).
}
 \begin{figure}[H]
 	\centering
 	\includegraphics[width=0.8\linewidth]{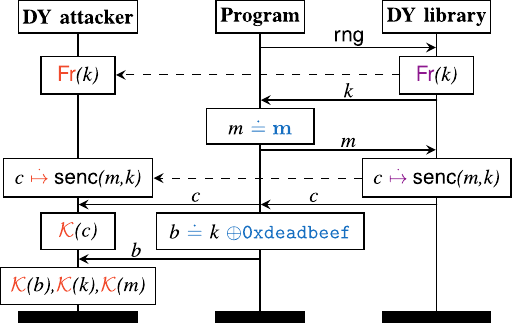} 
 \end{figure}
\end{example}

With an empty deduction combiner, the masked bitstring in the second
network message
is only accessible via
the symbol $b$. The DY attacker can perform DY
operations on the symbol $b$, but there is no way to access the $k$
symbol without reasoning about the bitstring. Hence the empty
deduction combiner under-approximates the adversaries' capabilities on
operating with bitstrings. This is equivalent to the view in 
\cite{sprengerIglooSoundlyLinking2020,arquint2022sound,arquint2023generic}, where the
concrete attacker is simply a translation function around the DY
attacker. If a bitstring that cannot be parsed is encountered, it can
only be ignored.

At the opposite end of the spectrum, Backes et al.~\cite{backes2016computational} aimed for computational
soundness, which entails that all attacks that could be mounted by
a Turing machine must be captured by the DY attacker. As the Turing
machine can reverse the $\oplus$ operation in the above example, this
required an over-approximation where all bitstring operations were
represented in the DY model as \emph{transparent} function symbols,
i.e., function symbols whose input parameters are fully accessible. 

We can generically represent this over-approximation in our framework,
if we have an equality predicate $\eqone$ in the program's
predicate set and we can identify the set of symbols that appear on
either side, say, using a function named $\mathit{symbols}$:
\begin{multline*}
	\label{comb-ded-over-approx}\tag{over-approx}
    \subone{\predset} \uplus \subtwo{\predset} \dedrelcomb^\top \knowledge{z} 
    \iff \exists \ x, y. \\
    \knowledge{x} \in \subtwo{\predset}
    \land
    \progEqual{x}{y} \in \subone{\predset}
    \land
    z \in \mathit{symbols}(y)
\end{multline*}
This over-approximation can introduce spurious equalities that lead to 
false attacks. For example, it is reasonable that a logic for
bitstrings can conclude $ \progEqual{a}{a \oplus x \oplus x}$ for any $a$
and $x$. This could easily introduce a spurious dependency between
some $a$ transmitted to the attacker
and an arbitrary symbol $x$. 

\subsubsection{Sharing equalities}\label{subsec:share-eq}

\newrevised{If we can identify at least one equality predicate in both component's predicate space,} however, we can find a much more useful
middle ground between both extremes \newtext{(i.e., over- and under approximation).} 
Connecting equality
judgments in both systems may allow tracking data flow across system
boundaries, while requiring nothing more than to point out the
equality predicates.\footnote{
This task could even be automated by (heuristically) identifying an
equality as a predicate of arity two that is symmetric, reflexive and
transitive.}

Let $\subone{\predSpace}$ and $\subtwo{\predSpace}$ contain atoms
$\eqone$ and $\eqtwo$ such that
symbols can appear on each side of
either of them, i.e., $x \doteq_i y \in \predSpace_i$ for $i\in\set{\colone,\coltwo}$
and $x,y\in \subone{\symbset} = \subtwo{\symbset}$.
Then we can transfer equalities with the
minimal deduction combiner defined by the following statements:
\begin{align*}
    \subone{\predset} \uplus \subtwo{\predset} \dedrelcomb^\mathsf{eq} \progEqual{x}{z}
    & \iff \exists y. \ 
    \progEqual{x}{y} \in \subone{\predset}
    \land
    \dyEqual{y}{z}\in \subtwo{\predset} \label{transfer-equalities-eq1}\tag{\eqone}
    \\
    \subone{\predset} \uplus \subtwo{\predset} \dedrelcomb^\mathsf{eq} \dyEqual{x}{z}
    & \iff \exists y. \ 
    \progEqual{x}{y} \in \subone{\predset}
    \land
    \dyEqual{y}{z} \in \subtwo{\predset}\label{transfer-equalities-eq2}\tag{\eqtwo}
\end{align*}
\newtext{The following example shows how we address the \emph{loss of bit-level information} discussed in~\autoref{lose-bit-info} where the DY attacker could not analyze cryptographic secrets with bit-level modifications. Even though the DY attacker still cannot directly analyze bitstrings, they can now leverage the program's analysis by transferring equivalences through equality combiners (\autoref{transfer-equalities-eq1} and~\autoref{transfer-equalities-eq2}). This is essential when the protocol implementation involves packing (i.e., formatting messages so that the other party on the network \newrevised{can} read them) and unpacking (i.e., extracting the message).}

\begin{example}[Transferable equalities]\label{ex:Transferable-equalities}
	Equality can easily be transferred to accrue logical deduction relations.
	\newrevised{A $ \progEqual{\!\!}{\!\!} $ predicate can be added to the program’s predicate set, e.g., similar to what will be discussed in~\autoref{subsec:bir}.}
	In the following procedure block, the deduction relation $\subone{\dedrel}$ is used to deduce $ \progEqual{k'''}{k'} $.
	Given $ \progEqual{k''' \!}{\! k'} $ and $ \dyEqual{k'}{k} $, the attacker infers $ \dyEqual{k'''}{k} $ using~\autoref{transfer-equalities-eq2}.
	Knowledge of $k'''$ is derived from $\knowledge{k}$ and $\dyEqual{k'''}{k}$ employing the \infruleref{fig:DYdeductionrules}{\rSubst} rule in~\autoref{fig:DYdeductionrules}.
	Consequently, the attacker learns $ \fstlangcol{m} $ by knowing $k'''$, $ c $, and \revised{\infruleref{Ra-w3.c}{RA,W3.c}}{$\dyAliEq{c}{\senc(m,k''')}$}.
\begin{figure}[H]
    \centering
    \includegraphics[width=0.8\linewidth]{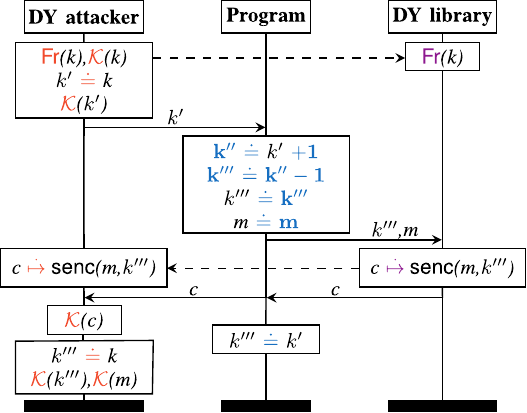}    
\end{figure} 
\end{example}

\subsubsection{Combined reasoning}
\label{comb-reason}
Equality sharing can transfer many statements derived from the other
component into the predicate space of the DY attacker, but
(a) only those
that discuss the relation between term sent or deduced
by the attacker (as only those have symbols), and, 
(b) only if the other
component has sufficient information to derive an equality judgment.

Coming back to~\Cref{ex:masked-encryption}, we see that the masking around the
encryption key must be removed to deduce ${k}$ from $b$.
\newtext{
But as the program does not perform that operation, the necessary equality
(between $k$ and the potential result of such an operation) is not produced.}
The \emph{ability} to perform this operation
must be described via the $\knowledge{\cdot}$ predicate rather than
$\eqtwo$.
A sound way of doing that would be to enhance the DY attacker with
bitstring manipulation via constant values.
\begin{multline*} 
	\label{comb-ded-bit}\tag{bit}
    \subone{\predset} \uplus \subtwo{\predset}
    \dedrelcombbit
    \knowledge{x} 
    \iff  \exists \ y, c. \  \\
    \knowledge{y} \in \subtwo{\predset}
    \land
    \progEqual{y}{\fstlangpred{op}(x,c)} \in \subone{\predset}
    \land
    \fstlangpred{const}~c \in \subone{\predset}
\end{multline*}
This combinator depends on the predicate space ${\subone{\predSpace}}$ providing 
a predicate
$\fstlangpred{const}~c$
that indicates a constant
and needs to explicitly list all binary operators
$\fstlangpred{op}(x,c)$.
It thus cannot be regarded as generic, although these concepts
(operators and constants) should apply to many programming
languages.
Again recalling~\Cref{ex:masked-encryption},
we can use $\dedrelcombbit$ to derive 
$\knowledge{k}$,
from
$\knowledge{b}$,
$\progEqual{b}{k \fstlangcol{ \oplus \mathtt{0xdeadbeef}}}$
and
$\fstlangpred{const}~\fstlangcol{\mathtt{0xdeadbeef}}$.

\begin{figure}
	\centering
	\includegraphics[width=0.8\linewidth]{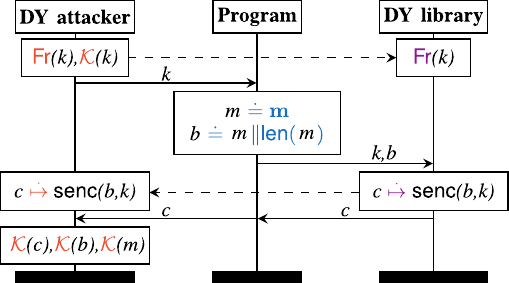} 
    \caption{A DY attacker removing  bit-level masking using $\dedrelcombbit$ in~\Cref{ex:BS-manipulation}}
	\label{fig:exp1}
\end{figure}
\newtext{
Similarly, when we come back to \Cref{ex:BS-manipulation}, in
\newrevised{\autoref{fig:exp1}} we can see how $\dedrelcombbit$ helps 
the DY attacker derive $\knowledge{\com{m}}$.
\newrevised{
As $ \fstlangpred{len}\fstlangcol{(\com{m})}$ is a constant and 
$\fstlangcol{\|}$ is an operation applied to $ \com{m} $ and 
$ \fstlangpred{len}\fstlangcol{(\com{m})}$, the DY attacker obtains $\knowledge{\com{m}}$ from $ \knowledge{\com{b}} $, $ \progEqual{\com{b}}{\fstlangcol{\com{m} \|}\fstlangpred{len}\fstlangcol{(\com{m})}} $ and $ \fstlangpred{const}~\fstlangpred{len}\fstlangcol{(\com{m})} $.
}
We have now addressed
the issue of parsing assumptions  (\autoref{sec:parsing})
in
\Cref{ex:BS-manipulation} and
the loss of bit-level information (\autoref{lose-bit-info})
in
\Cref{ex:masked-encryption}.
}

In summary,
the symbolic view on composition improves the accuracy of judgment in particular when combining
with the DY attacker as~\Cref{ex:BS-manipulation,ex:masked-encryption,ex:Transferable-equalities} witness.
This is hardly surprising,
as the translation approach 
sets up both DY attacker and program in a concrete execution
semantics with concrete (classical) composition, although the DY
attacker is symbolic in nature. 
By instead lifting the language to the symbolic level, we
turn the composition approach back on its feet and observe---at the
level of the composed system---that we have two methods of deduction
at our disposal.
What is surprising, is that we can achieve a significant improvement with
relatively simple deduction combinators. It should be difficult to find
logics where one \emph{cannot} find an equality predicate. Even a closer
integration as sketched in the previous paragraph, would apply to
a large set of programming languages while yielding immediate benefits.

\subsection{Beyond DY Attackers}
\label{subsec:Beyond-DYA}
Besides the DY model,
which is used in protocol verification,
there are two other attacker models that we want to discuss in the context of this framework.
The first is the \emph{unbounded attacker} used in programming languages and system-level verification.
This attacker is used in settings where cryptographic primitives are either not used at all,
or where their security guarantees are built into the language
semantics~\cite{lincoln1998probabilistic}.
The unbounded attacker can be a program or program context in the same
language as the program under verification,
or 
the trace of inputs that the program interacts with.
In both cases, computational limitations (even decidability) are
rarely relevant to the security argument. 
The decoupling of the attacker is thus only interesting if the attacker is intended to communicate with \newrevised{multiple other} components.
In
this case, there is no need for deduction by the unbounded attacker (each of its inputs is arbitrary, so fresh symbols)
but deduction combiners can be useful for components that share information, e.g., via the attacker.

The second attacker model is the \emph{computational attacker}, mentioned in \autoref{sec:comp-attacker}. There, we discussed how the translation approach struggles with probabilistic choice, unless the language provides the means to draw random keys. 
A naive formulation of the computational attacker encodes the Turing
machine semantics or any other probabilistic semantics. E.g.,
when a key is drawn, there are approximately 
$2^\secpar$ possible next states, with $\secpar$ being the key length, each describing a different value of
this key after sampling.
It is clear that such a modeling has little use for verification, as the state space is enormous.

Instead, we can apply our previous argument that a symbolic semantics
for the program ought to be composed with a symbolic semantics for the attacker and library. Thus, we should find a \newrevised{symbolic} representation of the random process producing, e.g., a distribution over keys. 
Bana and Comon propose a model where
symbolic rules with  a computational interpretation are
individually proven sound, but can be used to reason symbolically~\cite{banaComputationallyCompleteSymbolic2014,
bana2012towards}.
It can be reasoned about 
interactively with the \Squirrel
prover~\cite{baeldeInteractiveProverProtocol2021}.
We only sketch the idea here and leave a full
realization for future work.
As for the DY attacker, the computationally-complete symbolic attacker
(CCSA) represents messages as terms over a set of function symbols and
names, however, they are interpreted w.r.t.\ a
security parameter. A name describes the process of sampling
a random bitstring. A term describes a recursive process of evaluating
each function symbol using some polynomial algorithm and sampling each
name as described (but only once). 
In contrast to the DY model, where the function symbols define what
the attacker can do (and everything else is disallowed), the CCSA
model retains compatibility with the computational model by
symbolically formulating what the attacker \emph{cannot do} (and
everything else is allowed).
Consequently,
there is no equational theory; equality
is evaluated literally on the resulting bitstrings (in the
interpretation). Instead, CCSA features axioms that are proven sound
w.r.t.\ the above mentioned interpretation of terms as probabilistic polynomial-time Turing machines. 
The CCSA is thus simpler to define than the DY attacker: it
does not retain a predicate set or an equality predicate. The
predicate set, however, is the first-order logic described by Bana and
Comon~\cite{banaComputationallyCompleteSymbolic2014}.
Scerri's decision procedure allows handling a fragment of these
formulas~\cite{scerriProofsSecurityProtocols2015}, hence there is even
potential for automation.

\subsection{Correctness}
\label{sec:correctness}
 \revised{\infruleref{Ra-w3.d1}{{\small RA,W3.d1}}}{
 	
 	The correctness of parallel composition ($\Parallel[]{}{}$) is defined in terms of 
a partially synchronized interleaving ($\interleaving{}{}$) of the traces of each component, i.e.,
a permutation of the union of trace sets that maintains the relative order of 
elements within each set (see \appendixorextended{sec:interleave} for the formal definition).}
This is stronger than trace inclusion; for instance, it implies that all non-synchronizing traces of \newrevised{$ \fstlangcol{\Tracesfun{M}} $ or $ \sndlangcol{\Tracesfun{M}} $
are contained in $  \interleaving{\fstlangcol{\Tracesfun{M}}}{ \sndlangcol{\Tracesfun{M}}} $, where $\Tracesfun{M}$ is the set of traces produced by an LTS $M$}.
\revised{\infruleref{Ra-w1.e}{RA,W1.e}}{
	\reversemarginpar
    The correctness result covers \emph{all} events, including
    synchronizing events for the
    DY attacker, the DY library and non-synchronizing events 
    that occur only in the program we translate.
    Verification methodology will typically only consider a specific
    subset. 
    In our case studies, for instance, the program emits
    non-synchronizing events when
    special functions are reached and the verification tool describes
    security properties as trace properties over these events.
}
\normalmarginpar

More precisely, a trace $ \Trace \in \Traces$ is a sequence of events. 
We denote the symbolic parallel composition by $ \SParallel[\cdot]{}{}$ and traditional parallel composition for concrete systems by $ \CParallel{}{} $.
To avoid any ambiguity, we use notation like
 $ \mixtrace \in  \mixtraces(\Parallel[]{\fstlangcol{M}}{\sndlangcol{M}})$ 
 to refer to the  
sequence of events produced by a composite system
and
$\symtraces $ to distinguish the set of symbolic semantics traces
from
the set of concrete semantics traces $ \contraces $.

\begin{theorem}[Symbolic Composition Correctness]
	\label{thm:symb-compthm}
	For any symbolic LTS $\fstlangcol{M}$ and $ \sndlangcol{M}$, and for any combined deduction relation $\dedrelcomb$:
	\begin{enumerate}
            \item \newrevised{If all predicates $ \dedrelcomb^\mathsf{ena}$ produces may enable additional transitions, but not disable them, we call  $ \dedrelcomb^\mathsf{ena}$ \emph{enabling} and $\sbirmixtraces(\SParallel[\dedrelcomb^\mathsf{ena}]{\fstlangcol{M}}{\sndlangcol{M}}) \supseteq \interleaving{\fstlangcol{\symtraces}(\fstlangcol{M})}{\sndlangcol{\symtraces}(\sndlangcol{M})} $.}
		\item \newrevised{If all predicates $ \dedrelcomb^\mathsf{dis}$ produces 
                may disable transitions, but never enable new transitions, we call 
                $ \dedrelcomb^\mathsf{dis}$ \emph{disabling} and
            $ \sbirmixtraces(\SParallel[\dedrelcomb^\mathsf{dis}]{\fstlangcol{M}}{\sndlangcol{M}}) \subseteq \interleaving{\fstlangcol{\symtraces}(\fstlangcol{M})}{\sndlangcol{\symtraces}(\sndlangcol{M})} $.}
	\end{enumerate}
(The interested reader may consult~\theappendixorfull{sec:transitions} for the formal definitions of transition enabling and disabling.)
\end{theorem}
\begin{proof}
    By induction over the length of the composed trace. The base case is trivial (no step is taken). The inductive case is proved by a case distinction over synchronous and asynchronous events.~\href{https://github.com/FMSecure/CryptoBAP/tree/main/HolBA/src/tools/parallelcomposition/generaldeduction/derived_rules_generaldeductionScript.sml#L586}{\itshape\underline{Correctness-Enable}} and~\href{https://github.com/FMSecure/CryptoBAP/tree/main/HolBA/src/tools/parallelcomposition/generaldeduction/derived_rules_generaldeductionScript.sml#L484}{\itshape\underline{Correctness-Disable}} 	\newtext{mechanize the proof of \autoref{thm:symb-compthm}'s cases in HOL4.}
\end{proof}

\autoref{thm:symb-compthm} enables compositional analysis of symbolic
systems, as \autoref{thm:symb-comprefin} shows. 
Let refinement (or security) be expressed in terms of trace inclusion. 
Then,
if 
component $ \fstlangcol{M_1} $ refines $ \fstlangcol{M_2} $, written
in the same language,
and the same holds for
components $ \sndlangcol{M_1} $ and $ \sndlangcol{M_2} $,
then the combined system
  $\SParallel[\dedrelcomb]{\fstlangcol{M_1}}{\sndlangcol{M_1}}$
refines
  $\SParallel[\dedrelcomb]{\fstlangcol{M_2}}{\sndlangcol{M_2}}$.

\begin{lemma}[\newrevised{Symbolic Compositional Trace Inclusion}]
	\label{thm:symb-comprefin}
    Let $ \fstlangcol{\symtraces_1} $, $ \fstlangcol{\symtraces_2} $, $ \sndlangcol{\symtraces_1} $, and $ \sndlangcol{\symtraces_2} $ be the sets of traces produced by 
	any symbolic LTS $\fstlangcol{M_1}$, $ \fstlangcol{M_2} $, $ \sndlangcol{M_1} $, and $ \sndlangcol{M_2} $. If $ \fstlangcol{\symtraces_1 \subseteq \symtraces_2} $ and $ \sndlangcol{\symtraces_1 \subseteq \symtraces_2} $, then:
	\begin{itemize}
		\item For the empty combined deduction relation $\emptyset$, it holds that \href{https://github.com/FMSecure/CryptoBAP/tree/main/HolBA/src/tools/parallelcomposition/deduction/derived_rules_deductionScript.sml\#L207}{$ \sbirmixtraces(\SParallel[]{\fstlangcol{M_1}}{\sndlangcol{M_1}})  \subseteq \sbirmixtraces(\SParallel[]{\fstlangcol{M_2}}{\sndlangcol{M_2}})$}.
		\item For any combined deduction relation $ \dedrelcomb \ \in \{\dedrelcomb^{\top},\dedrelcomb^\mathsf{eq},\dedrelcombbit\}$ (defined in~\autoref{subsec:DedComb}), it holds that
		\href{https://github.com/FMSecure/CryptoBAP/tree/main/HolBA/src/tools/parallelcomposition/combinededuction/derived_rules_combinedeductionScript.sml\#L177}{$\sbirmixtraces(\SParallel[\dedrelcomb]{\fstlangcol{M_1}}{\sndlangcol{M_1}})\subseteq\sbirmixtraces(\SParallel[\dedrelcomb]{\fstlangcol{M_2}}{\sndlangcol{M_2}})$}.
		\item For any disabling combined deduction relation $\dedrelcomb^\mathsf{dis}$ on the refined system (left) and any enabling combined deduction relation $\dedrelcomb^\mathsf{ena}$ on the abstract system (right), it holds that
		\href{https://github.com/FMSecure/CryptoBAP/tree/main/HolBA/src/tools/parallelcomposition/generaldeduction/derived_rules_generaldeductionScript.sml\#L601}{$ \sbirmixtraces(\SParallel[\dedrelcomb^\mathsf{dis}]{\fstlangcol{M_1}}{\sndlangcol{M_1}})  \subseteq \sbirmixtraces(\SParallel[\dedrelcomb^\mathsf{ena}]{\fstlangcol{M_2}}{\sndlangcol{M_2}}) $}.
	\end{itemize}
\end{lemma}

In~\theappendixorextended{sec:dy-lib-comp}, we instantiate~\autoref{thm:symb-compthm} to enable merging and splitting DY libraries containing the same or distinct function signatures, as protocol parties often utilize different implementations for crypto libraries.

\subsection{Refinement}\label{subsec:refinement}

While \autoref{thm:symb-compthm} and \autoref{thm:symb-comprefin} are used
throughout our proof in \autoref{sec:case-study}, we need an additional
theorem to carry the analysis to the concrete system
semantics. This follows from the fact that both theorems only make statements about 
symbolic semantics traces ($\symtraces $) instead of
concrete semantics traces ($\contraces$). We, thus, need to relate the two.

Symbolic execution semantics are usually defined sound using
a refinement relation, \revised{\infruleref{Rb-m1}{RB,M1}}{which we denote as $\sqsubseteq$.
To define it, we have to assume that we have a way to apply 
a (component-specific) \emph{interpretation function},
i.e., a function $ \interpret$  from symbolic variables to concrete values~\newrevised{(e.g., the function $ H $ utilized in~\cite{lindner2023proofproducingsymbolicexecutionbinary})},
to a symbolic trace. E.g., let $\fstlangcol{\interpreterapply}$
denote this application,
$\fstlangcol{\contrace \sqsubseteq \symtrace \iff \exists \interpret .  \ \contrace
= \interpreterapply(\symtrace, \interpret)}$ and likewise for $\sndlangcol{\sqsubseteq}$.}
With this notation, we describe how refinement transfers to
the composed system.

\begin{theorem}[Refinement]
	\label{thm:traceinc}
	For any \newrevised{enabling} combined deduction relation $\dedrelcomb^\mathsf{ena}$,
	any concrete LTS $\fstlangcol{M_c}$ and $ \sndlangcol{M_c} $, 
	any symbolic LTS $ \fstlangcol{M_s} $ and $ \sndlangcol{M_s} $, 
	we have
	\begin{center}
		\begin{prooftree}	
			\hypo{\begin{matrix}
					\fstlangcol{\contraces}(\fstlangcol{M_c}) \ \fstlangcol{\sqsubseteq} \ \fstlangcol{\symtraces}(\fstlangcol{M_s})
			\end{matrix}}
			\hypo{\begin{matrix}
					\sndlangcol{\contraces}(\sndlangcol{M_c}) \ \sndlangcol{\sqsubseteq} \ \sndlangcol{\symtraces}(\sndlangcol{M_s})
			\end{matrix}}
			\infer2[]{ \conmixtraces(\CParallel{\fstlangcol{M_c}}{\sndlangcol{M_c}}) 
                           \sqsubseteq
             \sbirmixtraces(\SParallel[\dedrelcomb^\mathsf{ena}]{\fstlangcol{M_s}}{\sndlangcol{M_s}})}
		\end{prooftree}
	\end{center}
       \revised{\infruleref{Rb-m1}{RB,M1}}{where 
        $ \sqsubseteq $
        is defined from 
        $\fstlangcol{\interpreterapply}$
        and
        $\sndlangcol{\interpreterapply}$ \newrevised{according to~\refappendixorfull{def:ref-comp} in~\theappendixorfull{sec:comp-refine}.}
    }
\end{theorem}

In~\autoref{thm:traceinc}, the enabling deduction relation is to ensure broader coverage of behaviors during symbolic execution compared to concrete execution.

\begin{proof}
        From the left-hand side, we apply a concrete variant of 
        \autoref{thm:symb-compthm} (\refappendixorfull{thm:conccompthm} in~\theappendixorfull{sec:conc-thms}) to describe the composed concrete system via
        interleaving.
        From the right-hand side, we apply \autoref{thm:symb-compthm} \newrevised{(case 1)} itself, 
        to obtain a similar interleaving, but of the composed
        symbolic system.
        We then use 
       \newrevised{the refinements $ \fstlangcol{\sqsubseteq} $, $ \sndlangcol{\sqsubseteq} $, and $ \com\sqsubseteq $ and instantiate the interpretation functions in $ \com\sqsubseteq $ to map the composed traces in the concrete domain to those in the symbolic domain.
    }
	See~\href{https://github.com/FMSecure/CryptoBAP/tree/main/HolBA/src/tools/parallelcomposition/refinement/refinementScript.sml#L66}{\itshape\underline{Refinement}} for proof mechanization in HOL4.
\end{proof}

We avoid communication between symbolic and concrete components
by avoiding hybrid systems altogether, which is why we need 
both \autoref{thm:symb-comprefin} (for abstraction within the symbolic domain)
and \autoref{thm:traceinc} (for abstraction from the concrete to the
symbolic domain).

\revised{\infruleref{Rb-m1}{RB,M1}}{
	The reader may wonder if the interpretation function $ \interpret $,	
	used in the definition of $ \sqsubseteq $,
	would also constitute a translation of the kind we criticized.
	We criticize the implication of a translation at the object level,
	i.e., translation \emph{within} the system between concrete
	programs and symbolic DY attackers \emph{in the same system}.
	By contrast, the interpretation function resides at the proof
	level rather than the object level, and (the existence of it) is
	merely a constraint that the symbolic execution is a consistent
	abstraction.
	Concretely, it can be created on the fly and it is not required to
	be computable or consistent across multiple (symbolic) executions.
}

\section{Instantiations of the Framework}
\label{sec:case-study}
We instantiate our framework with different languages: 
(a) ARMv8 and RISC-V for verifying implementations of real-world protocols, 
(b) \Sapic for modeling parties from the specification, and 
(c) DY rules for specifying our threat model.
\revised{\infruleref{Rc-w1}{RC,W1}}{
We will consider a case study (WireGuard, see below)
where we extract both parties, client and server, from ARMv8 binaries.
We will also consider a case study (TinySSH), where the ARMv8 binary
describes only the server, and the client behavior is described in
\Sapic, thereby modeling an unknown SSH client implementation that
follows the specification~\cite{rfc4253}. Both cases include a DY
attacker, and the second mixes a machine-code language (case a) with
a specification language (case b).
Both cases are relevant, as protocol standards may not always be
available, for instance, if the protocol is not widely used or not yet
standardized. Vice-versa, protocol standards can be ambiguous or
overly-general, so it can be interesting to consider a particular
implementation. 
}

To derive the \Sapic model of the protocols' parties from their binary implementations, we transform their machine code into \birsymb{}, symbolically execute them, and translate the resulting execution trees into (a subset of) \Sapic.
This section demonstrates how the theorems presented so far simplified the end-to-end proof, enabling us to mechanize it.
Finally, we \revised{\infruleref{Rc-m3}{RC,M3}}{prove mutual authentication and forward secrecy in the symbolic model for} the TinySSH and WireGuard protocols to evaluate our framework.

\subsection{Intermediate representations}
\subsubsection{\bf The \birsymb{} Representation}
\label{subsec:bir}
\begin{figure}
\adjustbox{varwidth=\linewidth,scale=0.95}{%
\begin{equation*} 
\begin{split}
	\birprog \in \fstlangcol{{prog}} & \bnfdef  \fstlangcol{{block}}^* \\
	\fstlangcol{{block}} & \bnfdef  ({\var{v},\fstlangcol{{stmt}}^*}) \\
	\var{v} \in \birval & \bnfdef  \mathit{string} \bnfsep \mathit{int}\\
	\fstlangcol{{stmt}}  & \bnfdef \fstlangcol{halt} \bnfsep \fstlangcol{jmp}(\var{e}) \bnfsep \fstlangcol{cjmp}(\var{e},\var{e},\var{e})\\  & \ \ \ \ \ \ \bnfsep \fstlangcol{assign}(\mathit{string}, \fstlangcol{e}) 
\\	\var{e} \in \birexp & \bnfdef  \var{v}\;  \bnfsep \fstlangcol{var} \ \mathit{string}\;  \bnfsep  \fstlangcol{\Diamond_u} \; \var{e}\; \bnfsep \var{e} \; \fstlangcol{\Diamond_b} \; \var{e}  \\ 
\end{split}
\end{equation*}
}
\caption{A fragment of the $\birsymb$ syntax}
\label{fig:birSyntax}
\end{figure}

We use HolBA~\cite{DBLP:journals/scp/LindnerGM19}---\newrevised{a binary} analysis platform in HOL4---to transpile the protocols' binary into the \emph{binary intermediate representation} ($\birsymb$). $\birsymb$ is a simple and architecture-agnostic language designed to simplify the analysis tasks and is used as the internal language of HolBA to facilitate building analysis tools.
The $ \birsymb $  transpiler is verified and generates a certifying theorem that guarantees that the semantics of the binary is preserved (see \cite[Thm.~2]{DBLP:journals/scp/LindnerGM19}); this ensures that the analysis results on $\birsymb$ can be transferred back to the binary.
\autoref{fig:birSyntax} shows the $ \birsymb $ syntax. A $\birsymb$ program $\birprog$ includes a number of blocks, each consisting of a tuple of a unique label (i.e., a string or an integer) and a few statements. 
The label of $ \birsymb $ blocks is often used as the target of jump instructions---\fstlangcol{jmp} or \fstlangcol{cjmp}---and refers to a particular location in the  program.
\revised{\infruleref{Ra-w1.b}{RA,W1.b}}{The $ \fstlangcol{assign} $ statement  assigns the evaluation of a $\birsymb$ expression to a variable,} and $\fstlangcol{halt}$ indicates the execution termination. 
$\birsymb$ expressions include constants, standard binary, and unary operators (ranged over by $\fstlangcol{\Diamond_b}$ and $\fstlangcol{\Diamond_u}$) for finite integer arithmetic.
$\birsymb$ expressions also include memory operations and conditionals, which we leave out to simplify the presentation and because they are unnecessary in our evaluation.

We use a proof-producing symbolic execution for $ \birsymb $~\cite{lindner2023proofproducingsymbolicexecutionbinary} that formalizes the symbolic generalization of $\birsymb$ (hereafter $ \sbirsymb $) to find all execution paths of the program.
The symbolic semantics aligns with the concrete semantics, enabling guided execution while ensuring a consistent set of reachable states from an initial symbolic state.
The set of $ \sbirsymb $ events is the disjoint union of the set of non-synchronizing events and the set of synchronizing events.
The set of synchronizing events encompasses $ \Syncfresh{n} $ for the secret name $ n \in \privnames $, $ \AtoP{x} $ and $ \PtoA{x} $ for the symbol $ x \in \symbset $, and $ \fcall{f,x_1,\ldots,x_n,y} $ for the function symbol $ f \in \Funcs^n $ and symbols $ x_1,\ldots,x_n,y \in \symbset$. 
The set of non-synchronizing events includes $ \event{e} $ to indicate the release of a visible event $ e $, $ \Loopsym $ to denote initiating a loop, \revised{\infruleref{Ra-w1.b}{RA,W1.b}}{and $ \assign{x}{\fstlangcol{e}}$ to signify assigning the $\birsymb$ expression $ \fstlangcol{e} $ to the symbol $ x $.}
Moreover, a sequence of events $ \sbirevent_{\fstlangcol{1}},\dots, \sbirevent_{\fstlangcol{m}} $ signifies a $ \sbirsymb $ trace $ \sbirTrace \in \sbirtraces $ such that $ \sbirTrace = \sbirevent_{\fstlangcol{1}},\dots, \sbirevent_{\fstlangcol{m}} $.

\input{sections/BIR-Exp4.tex}
\revised{\infruleref{Ra-w1.b}{RA,W1.b}}{
	\autoref{fig:bir-exp4} illustrates a sequence of $ \birsymb $ statements of the program in~\Cref{ex:masked-encryption}. Note that the  $ \birsymb $ representation is simplified w.r.t.\  to the implementation in HolBA.
	When our symbolic execution engine evaluates each of these $ \birsymb $ statements, a logical predicate may be added into the $ \sbirsymb $ predicate set, denoted as $ \sbirpred $, and an $ \sbirsymb $ event arises \newrevised{(i.e., in the exact way that the \CryptoBap{} performs symbolic execution~\cite[Sec.\ 5]{nasrabadi2023cryptobap}).}
	For example, an equality predicate, represented as $ \progEqual{\!\!}{\!\!} $, is added to the $ \sbirsymb $ predicate set as a result of processing the $ \fstlangcol{assign} $ statement. 
	Additionally, the DY attacker’s predicate set (i.e., $ \knowledgeset $) is updated due to the combined deduction relation $ \dedrelAvsL $ and synchronization (\autoref{tab:events} summarizes synchronization events).} 
	  \revised{\infruleref{Ra-w2}{RA,W2}}{The parallel composition of \sbirsymb{} and the DY attacker employs the combined deduction relation $ \dedrelcombSbirDY $, which represents a specialized variant of the combined deduction relation $ \dedrelcombbit $ for \sbirsymb{}, as presented below:            
		\begin{multline*} 
			\label{comb-ded-sbir-dy}\tag{bit$'$}
			\sbirpred \uplus \knowledgeset
			\dedrelcombSbirDY
			\knowledge{z} 
			\iff  \exists \ x, y, w. \\
			\knowledge{y} \in \knowledgeset
			\land
			(\progEqual{y}{x \ \fstlangcol{\Diamond_b} \ w}) \in \sbirpred \ 
			\land  \\
			z \in \big(\mathit{symbols}(x) \cup \mathit{symbols}(w)\big)
	\end{multline*}                
	\newrevised{As shown in the last column of ~\autoref{fig:bir-exp4},} the DY attacker gains further logical facts by using the deduction combiner $ \dedrelcombSbirDY $ together with the DY and $ \sbirsymb $ predicate sets. \newrevised{We use a number next to each piece of knowledge, to indicate in which order these facts are acquired in our running example. }
}

\subsubsection{\bf The \Sapic{} Representation}\label{sec:sapic+}
\begin{figure}[t]
\adjustbox{varwidth=\linewidth,scale=0.95}{%
\centering	
	\[
	\begin{array}{ll ll}
		
		\langle \sndlangcol{P},\sndlangcol{Q}\rangle~::= \\[2mm]
		~~~~~  \sndlangcol{0}  &
		~~|~~\sndlangcol{!P}\\
	
		~~|~~\pin(x);~\sndlangcol{P}  &
		~~|~~\sndlangcol{P \mid Q} \\
		
		~~|~~\pout(x);~\sndlangcol{P} &
		~~|~~\pndc{\sndlangcol{P}}{\sndlangcol{Q}} \\
			
		~~|~~\pevent~{e};~\sndlangcol{P} & 
		~~|~~\pnew~n;~\sndlangcol{P} \\
	
		~~|~~ \plet~t_1 = t_2~\pin~\sndlangcol{P}~\pelse~\sndlangcol{Q} &
		
	\end{array}
	\]
 }
	\caption{A fragment of the syntax of \Sapic process calculus.
		In this figure,
	$e, t_1,t_2 \in \Term$, 
	$n \in \privnames$, 
	$x \in \Vars$.
}
	\label{fig:SAPICsyntax}
\end{figure}

\Sapic is an applied pi calculus similar to \ProVerif and \SapicOriginal~\cite{kremer2016automated}.
\Sapic extends \SapicOriginal with the definition of destructors, $ \plet $ bindings with pattern matching and $ \pelse $ branches.
\Sapic provides a common language that (soundly) translates to both \ProVerif and
\Tamarin~\cite{cheval2022sapic+}.
\autoref{fig:SAPICsyntax} presents a fragment of the \Sapic{'s} syntax.
The $\pnew$ construct creates fresh values, and 
$\pin$ and $\pout$ receives and sends messages over the channel.\reversemarginpar
\revised{\infruleref{Ra-w1.e}{RA,W1.e}}{
The $ \pevent $ construct raises events that security
properties can
refer to, but otherwise does not change the execution. 
They will be used to capture event functions.
}\normalmarginpar
\revised{\infruleref{Ra-w2}{RA,W2}}{
    \Sapic contains the non-deterministic choice
    operator, denoted as $\pndc{\!\!}{\!\!}$ and introduced
    in~\cite{backes2017novel}.
    A process $\sndlangcol{P + Q}$ can either move as if it were $\sndlangcol{P}$, or as if it
    were $\sndlangcol{Q}$. 
}
$\Sapic$ syntax includes \textit{stateful} processes ~\cite{cheval2022sapic+} that manipulate globally shared states, i.e., some database, register or memory that can be read and altered by different parallel threads.
As we skipped the model extraction of memory manipulation primitives, the stateful processes are omitted in~\autoref{fig:SAPICsyntax}.

\SapicOld has the same syntax as 
\Sapic, but its semantics remove the DY attacker:
instead of invoking \Sapic's internal DY deduction relation,
communication (\pin, \pout) in \SapicOld emits events ($\AtoPsym$, $\PtoAsym$)
that synchronizes with an outside attacker.
\revised{\infruleref{Ra-w3.b}{RA,W3.b}}{
    Since \Sapic uses the event $\Kevent$ to signify messages coming
    from the attacker instead of
    $ \AtoPsym$, we use $\tracetrans{\cdot}$ to translate between trace
(sets) of $\Sapic$ and $\SapicOld$/\sbirsymb.
    Besides $\Kevent$, security properties in \Sapic can only refer to
    events in the process, which $\tracetrans{\cdot}$ keeps the same.
For a given \SapicOld process $\sndlangcol{P}$, $\sapictraces(\sndlangcol{P})$ denotes the set of all possible symbolic traces generated by process $ \sndlangcol{P} $.}
We define a \SapicOld trace $ \sapicTrace \in \sapictraces $ as a sequence of events such that $ \sapicTrace = \sapicevent_{_\sndlangcol{1}} \dots \sapicevent_{_\sndlangcol{m}} $.
In \autoref{sec:end-to-end},
we combine $\SapicOld$ with the DY attacker and library to \Sapic.
As a by-product, this
shows the correctness of both w.r.t.\ the DY 
semantics in \Sapic (which are \newrevised{quite} standard).

\subsection[From SBIR To SAPIC-]{From \sbirsymb{} To \SapicOld}
\label{subsec:sbir-sapic}
\begin{figure}
\adjustbox{varwidth=\linewidth,scale=0.92}{%
\begin{equation*} 
	\begin{aligned}[t]
\begin{split}
  \tree& = \!\leafnode \!\bnfsep \!\node(\birpc,\nodeevent)\bnfconcat\tree'\!\bnfsep \branchingnode(\birpc,\nodecond, \fstlangcol{T_1}, \fstlangcol{T_2}) \ \text{event tree} \\
  & \sbirtoiml{\leafnode} \hspace{1.4cm} \mapsto\ \hspace{0.1cm} \sndlangcol{0}\\
  & \sbirtoiml{\node(\birpc,\nodeevent)\ \bnfconcat\tree'} \bnfdef \hspace{1.4cm} \text{events nodes} \\
  &
 \begin{array}{lll}
  \sbirtoiml{\node(\birpc,\event{{e}})\bnfconcat\tree'} & \mapsto & \pevent \ {e};\sbirtoiml{\tree'}\\
  \sbirtoiml{\node(\birpc,\AtoP{x})\bnfconcat\tree'} & \mapsto & \pin (x);\!\sbirtoiml{\tree'}\\
  \sbirtoiml{\node(\birpc,\PtoA{x})\!\bnfconcat\!\tree'}\! & \mapsto & \pout(x);\!\sbirtoiml{\!\tree'}\\
 \sbirtoiml{\node(\birpc, \\ \fcall{f,x_1,\ldots,x_n,y})\bnfconcat\tree'} & \mapsto & \plet \ y=f(x_1,\ldots,x_n) \\ && \pin \ \sbirtoiml{\tree'} \ \pelse\ \sndlangcol{0} \\
  \sbirtoiml{\node(\birpc, \assign{x}{\fstlangcol{e}})\bnfconcat\tree'} & \mapsto & \plet \ x = \sbirtoiml{\fstlangcol{e}} \ \pin \ \sbirtoiml{\tree'}  \\
  \sbirtoiml{\node(\birpc,\Syncfresh{n})\bnfconcat\tree'} & \mapsto & \pnew \ {n};\ \sbirtoiml{\tree'}\\
   \sbirtoiml{\node(\birpc,\Loopsym)\bnfconcat\tree'} & \mapsto & \sndlangcol{!} \ \sbirtoiml{\tree'}\\
\end{array}\\
  & \sbirtoiml{\branchingnode(\birpc,\nodecond, \fstlangcol{T_1}, \fstlangcol{T_2})}  \hspace{0.91cm}  \mapsto \hspace{0.3cm}  \pndc{\sbirtoiml{\fstlangcol{T_1}}}{\sbirtoiml{\fstlangcol{T_2}}} \\
  & \sbirtoiml{\var{e} \in \birexp} \bnfdef \hspace{2.8cm} \text{\birsymb{} expressions} \\
  &
 \begin{array}{lll}
  \sbirtoiml{\var{b} \in \birval}  & \mapsto & \constrans{\com{b}} \in \pubnames \\
  \sbirtoiml{\fstlangcol{var} \ x }  & \mapsto & {x}\in \Vars\\
  \sbirtoiml{\fstlangcol{\phi_1 \Diamond_b \phi_2 }}  & \mapsto & \sbirtoiml{\fstlangcol{\phi_1}}\sbirtoiml{\fstlangcol{\Diamond_b}}\sbirtoiml{\fstlangcol{\phi_2}} \hspace{0.6cm} \text{Binary operations}\\
  \sbirtoiml{\fstlangcol{\Diamond_b}} &\mapsto & \left\{\begin{array}{ll}                                     
                                         =   & \fstlangcol{Equal} \\
                                          \functionsymbol{plus}, \functionsymbol{mult}, \ \dots & \fstlangcol{Plus}, \fstlangcol{Mult}, \ \dots
                                        \end{array}
                                 \right.\\
  \sbirtoiml{\fstlangcol{\Diamond_u \phi'}}          & \mapsto & \sbirtoiml{\fstlangcol{\Diamond_u}}\sbirtoiml{\fstlangcol{\phi'}} \hspace{1.4cm} \text{Unary operations}\\
  \sbirtoiml{\fstlangcol{\Diamond_u}} &\mapsto & \left\{\begin{array}{ll}
                                        \neg            & \fstlangcol{Not} \\ 
                                        \bot            & \text{otherwise}
                                        \end{array}
                                 \right.     \\                        
\end{array}    \\   
\end{split}
\end{aligned}
\end{equation*}
}
\caption{Translating execution tree $\tree$ to $\SapicOld$ model:  $ {e}, {x}, x_1,\ldots,x_n,y \in \symbset $ are symbols, $ n \in \privnames $ is a secret name, and $ \functionsymbol{plus}, \functionsymbol{mult}, f \in \Funcs^n $ are function symbols. Observe that $  \neg ({t_1} = {t2} ) = (\Notequal{t_1}{t_2}) $ for $  t_1,t_2 \in \Term $.
} 
\label{fig:sbirtosapic}
\end{figure}

Using symbolic execution, we derive the execution tree $ \tree $ of a $ \birsymb $ program, which is used to extract the $ \SapicOld $ model. 
The leaves in $ \tree $ are due to the $ \birsymb $ $ \fstlangcol{halt} $ statement that marks the end of a complete path.
A node in $ \tree $ is either a branching node 
$\branchingnode(\birpc, \nodecond, \tree_{\fstlangcol{1}}, \tree_{\fstlangcol{2}})$, where  $\birpc $ locates the conditional statement in the program, $\nodecond$ is the condition, and $\tree_{\fstlangcol{i}}$ are the sub-trees for $ \fstlangcol{i} \in \{\fstlangcol{1}, \fstlangcol{2}\}$; or 
an event node $\node(\birpc, \nodeevent)\bnfconcat\tree'$ with the sub-tree $ \tree' $ and $ \birpc $ specifying where the event $ \nodeevent $ occurred. 
In $\tree$, an edge connects two nodes if they are in the transition relation.

We construct $ \tree $ from a $\birsymb$ program and an initial symbolic state, with the root representing the initial state.
The symbolic execution provides us with up to two successor states for any node.
We obtain two successor states if the node represents a branching statement (i.e., \fstlangcol{cjmp}).
In such cases, the condition of the statement is stored in a branching node, and we proceed to construct subtrees from the two successor states.
If the node represents any other statement,
an event node is recorded with one or no successor tree.

The protocol model is obtained by translating $\tree$ into its $ \SapicOld $ model.
We translate $\tree$ 
using the rules in~\autoref{fig:sbirtosapic}
to 
$\sbirtoiml{\tree}$.
We translate leaves into a \textit{nil} process $\sndlangcol{0}$, and the event $ \nodeevent $ from the event nodes into their corresponding $ \SapicOld $ construct.
\revised{\infruleref{Ra-w2}{RA,W2}}{
The branching nodes of $\tree$ (i.e., $\branchingnode(\birpc, \nodecond, \tree_{\fstlangcol{1}}, \tree_{\fstlangcol{2}})$) are translated into a non-deterministic choice ($\pndc{\!\!}{\!\!}$) between (the translation of) both possible paths.
If these branches are not already pruned by symbolic execution, there
might still be bit-level conditions that are relevant for the protocol
verifier, but not sufficient to prune the branch during symbolic
execution. 
While we did not encounter this case, it would be possible to translate to
($\pndc{\pevent \ E_1; \sbirtoiml{\fstlangcol{\tree_1}}}{\pevent \ E_2; \sbirtoiml{\fstlangcol{\tree_2}}}$) for $E_1,E_2$ some unique events.
As \Sapic supports restricting the trace set based on formulas, we 
can reflect necessary conditions that are expressible in these tools.
For instance, if the condition was $x\fstlangcol{\oplus}y\fstlangcol{=\!0}$, we may require that the occurrence of $E_1 $,
i.e., a traversal into the positive branch, entails that $y
\sndlangcol{\neq}  x \sndlangcol{+1}$, if that helps exclude
a false attack.
In all our case studies, most paths are pruned by our symbolic execution
engine and the remaining not require such a refinement.
Nevertheless,
this feature would be easy to add (and prove correct for any
condition entailed by the combined deduction relation).
}

\revised{\infruleref{Rc-w5}{RC,W5}}{In order to illustrate the methodology employed for model extraction, the extracted \SapicOld process of the \birsymb{} program from Example~\autoref{ex:masked-encryption} is presented in~\autoref{fig:bir-exp4}.}
\revised{\infruleref{Rc-w3}{RC,W3}}{For technical details about the lifting of the binary and the symbolic execution, we refer to the \CryptoBap~\cite{nasrabadi2023cryptobap} paper, which has introduced this method for model extraction from the binaries, but without a mechanized proof. 
Our focus here is the end-to-end proof, which builds on the framework from
the previous section,
and the wider range of target backends provided by \Sapic.
}

\subsection{Translation correctness}
\label{subsec:theorems}

To enable transferring verified properties from the $ \Sapic $ level back to $\birsymb$ and then to the protocols' binary, it is essential to prove that the extracted $ \SapicOld $ model preserves the behaviors of the $\sbirsymb$ representation.
To this end, we establish a proof that for every path in the symbolic execution tree $ \tree $, there exists an equivalent $ \SapicOld $ trace derived from executing translated process $\sbirtoiml{\fstlangcol{\tree}} $.

\begin{theorem}[Trace Inclusion]
	\label{thm:traceinc:sbir-sapic}
\revised{\infruleref{Ra-w3.b}{RA,W3.b}}{	
    Let $\fstlangcol{\tree}$ be a $\sbirsymb$ execution tree. 
    Then,
	all translated $\sbirsymb$ traces of $\fstlangcol{\tree}$, $\tracetrans{\sbirtraces(\fstlangcol{\tree})} $, are included 
    in the traces of the translated $ \SapicOld $ process $\sapictraces(\sbirtoiml{\fstlangcol{\tree}})$.
}
\end{theorem}
\begin{proof}
	\reversemarginpar
	The proof is done by induction on the length of \revised{\infruleref{Ra-w3.b}{RA,W3.b}}{the translated traces $\tracetrans{\sbirtraces(\fstlangcol{\tree})} $}. 
	\normalmarginpar
        In the base case, no actions are taken.
	For the inductive case, we apply the case distinction over synchronous and asynchronous events in the set of $ \sbirsymb $ events.	
	We mechanized~\autoref{thm:traceinc:sbir-sapic}'s proof in HOL4  (see~\href{https://github.com/FMSecure/CryptoBAP/tree/main/HolBA/src/tools/parallelcomposition/translateTosapic/translate_to_sapicScript.sml#L493}{\itshape\underline{Symbtree-to-Sapic}}).
\end{proof}
Lindner et al.~\cite[Thm.~4.1]{lindner2023proofproducingsymbolicexecutionbinary} demonstrated that verified properties for $ \sbirsymb $ transfer to $ \birsymb $, ensuring that the verified properties hold for concrete execution semantics.

\subsection{End-to-end correctness result}
\label{sec:end-to-end}
We then show how \newrevised{the} theorems in~\autoref{sec:parallel} come together to simplify the analysis of our target language, which we will equip with
DY semantics. 
Our analysis below includes embedded links to the mechanized proof for each step.
We start with the concrete, complete ARMv8 program in parallel with an
unspecified attacker $A$. 
\begin{align*}
 &	 \contraces (\CParallel{C^{ARMv8}}{A})
\end{align*}
As is often the case, we take a detour via an intermediary language, in our case, $\birsymb$. \cite[Thm.~2]{DBLP:journals/scp/LindnerGM19} 
justifies this so-called \emph{lifting} step, i.e., shows that this translation is semantics preserving.
		Thanks to \refappendixorfull{thm:con-comprefin} in~\theappendixorextended{sec:conc-thms}, we can use this theorem in 
		context with $A$.
\begin{align*}
      \text{\href{https://github.com/FMSecure/CryptoBAP/tree/main/HolBA/src/tools/parallelcomposition/instantiations/arm8\_vs\_bir\_comp\_attackerScript.sml\#L145}{$=	\contraces (\CParallel{{C^{\birsymb}}}{A})$}}
\end{align*}
		We next require (Assumption 1) that $C^{\birsymb}$ is trace-equivalent to
		$\CParallel{P^{\birsymb}}{L^{\birsymb}}$, i.e., that it can be split into a 
		program-under-verification and a known library.
		\cite[Sec.\ 4]{nasrabadi2023cryptobap} provides statically checkable criteria for $\birsymb$
		to verify this condition automatically. Again, \refappendixorfull{thm:con-comprefin} is used to apply this in context.
		Afterwards, we
		use the refinement theorem \autoref{thm:traceinc}
		and the relations indicated by the underbraces to move
		from the concrete to the symbolic. 
		 An interpretation function $ {\fstlangcol{\interpret}} $ evaluates $ \sbirsymb $ symbolic expressions to $ \birsymb $ concrete values, as demonstrated in~\cite{lindner2023proofproducingsymbolicexecutionbinary}.
		Because $\CParallel{}{}$ is associative w.r.t.\ trace
		equivalence, we have:		
\begin{align*}		
	& = \contraces (\CParallel{\underbrace{P^{\birsymb}}_{\stackrel{\text{\cite[Thm.~4.1]{lindner2023proofproducingsymbolicexecutionbinary}}}{\fstlangcol{\sqsubseteq}}}}{\underbrace{\CParallel{L^{\birsymb}}{A}}_{\stackrel{\text{A2: Deduction Soundness}}{\sndlangcol{\sqsubseteq}}}}) \\
     &     \text{\href{https://github.com/FMSecure/CryptoBAP/tree/main/HolBA/src/tools/parallelcomposition/instantiations/bir_comp_attacker_vs_sbir_comp_DYScript.sml\#L67}{$\sqsubseteq \symtraces (\SParallel[{\dedrelcombSbirDY}]{\overbrace{P^{\sbirsymb}}}{\overbrace{\SParallel[{\dedrelAvsL}]{L^{DY}}{A^{DY}}}})$}} 
\end{align*}
The first relation is the soundness of symbolic
		execution.
                The second is an assumption on the attacker that
		we will talk about in a second. 
                Recall that $\SParallel[\revised{\infruleref{Ra-w3.c}{RA,W3.c}}{\dedrelAvsL}]{}{}$ 
                uses a deduction combiner specific to the DY attacker
                and library, 
                \revised{\infruleref{Ra-w2}{RA,W2}}{
                while $\SParallel[\dedrelcombSbirDY]{}{}$ utilizes a specialized deduction relation between \sbirsymb{} and DY as defined in~\autoref{subsec:bir}.
            Next, $\SParallel[\dedrelcombSapicDY]{}{}$ employs a deduction relation similar to $ \dedrelcombSbirDY $, referred to as $ \dedrelcombSapicDY $, which particularly applies to \SapicOld and DY predicate sets. The distinction from $ \dedrelcombSbirDY $ lies in the fact that $ \dedrelcombSapicDY $ incorporates the translation of the binary operators $\fstlangcol{\Diamond_b}$ as function symbols (see~\autoref{fig:sbirtosapic} for the translation of the binary operations).}
		We use
		\autoref{thm:symb-comprefin} \newrevised{(case 3)} to apply our
		translation result from $ \sbirsymb $ to \SapicOld (\autoref{thm:traceinc:sbir-sapic})
		 that we showed in the previous
		subsection (note that $P^{\SapicOld} = \sbirtoiml{P^{\sbirsymb}}$\newrevised{, $ \dedrelcombSbirDY $ is disabling, and $ \dedrelcombSapicDY $ is enabling}).
                We have:
\begin{align*}            
& \text{\href{https://github.com/FMSecure/CryptoBAP/tree/main/HolBA/src/tools/parallelcomposition/instantiations/sbir_sapic_comp_DYScript.sml\#L122}{$\newrevised{\subseteq} \ \symtraces (\SParallel[\dedrelcombSapicDY]{P^{\SapicOld}}{\SParallel[{\dedrelAvsL}]{L^{DY}}{A^{DY}}})$}}
\end{align*}
\Sapic 's semantics include the DY attacker and library, hence, the above system.
\begin{align*} 
\text{\href{https://github.com/FMSecure/CryptoBAP/tree/main/HolBA/src/tools/parallelcomposition/instantiations/sapic_comp_DY_sapicplusScript.sml\#L602}{$= \symtraces (P^{\Sapic})$}}
\end{align*}
We thus have an end-to-end correctness result.
Thanks to the framework theorems, this proof can be adapted to many
other languages, as the researcher needs to only show the
correctness of the language-specific steps (correctness of lifting, splitting, and
translation) when adapting. Moreover, they only need to be shown in
isolation. Until now, the translation step was usually shown in the presence of the adversary~\cite{kremer2016automated,aizatulinVerifyingCryptographicSecurity2015,nasrabadi2023cryptobap}.

While Assumption~1 delineates the class of programs that is supported
(and can be checked statically), Assumption~2 (short: A2) formalizes our 
threat model: whatever type of
system the attacker controls, it can be abstracted as a DY attacker if
we also abstract the library in the same way.
We can leave it at that, but we believe that this assumption merits deeper exploration, as discussed further in~\theappendixorfull{rel-comp-sound}.

In~\theappendixorextended{sec:multi-party}, we extend this proof to two parties (client and server) and an unbounded number of copies thereof.

\begin{table*}
	\centering
	\resizebox{2.1\columnwidth}{!}{
		\begin{tabular}{ll|c|c|c|c|c|c|c|ccc|c|c|c}
			\multicolumn{2}{c|}{\multirow{2}{*}{Protocol}} & \multirow{2}{*}{\begin{tabular}[c]{@{}c@{}}ARM\\ Loc\end{tabular}} & \multirow{2}{*}{\begin{tabular}[c]{@{}c@{}}Verified\\  Code Size\end{tabular}} & \multirow{2}{*}{\begin{tabular}[c]{@{}c@{}}Feasible\\ Path\end{tabular}} & \multirow{2}{*}{\begin{tabular}[c]{@{}c@{}}Infeasible\\ Path\end{tabular}} & \multirow{2}{*}{\begin{tabular}[c]{@{}c@{}}$\SapicOld$\\ Loc\end{tabular}} & \multirow{2}{*}{\begin{tabular}[c]{@{}c@{}}\Tamarin\\ Loc\end{tabular}} & \multirow{2}{*}{\begin{tabular}[c]{@{}c@{}}\ProVerif\\ Loc\end{tabular}} & \multicolumn{4}{c|}{Time (seconds)} & \multirow{2}{*}{Verified in} & \multirow{2}{*}{Primitives} \\
			\multicolumn{2}{c|}{} &  &  &  &  &  &  &  &  \multicolumn{1}{c|}{$\sbirsymb$} & \multicolumn{1}{c|}{$\SapicOld$} & \multicolumn{1}{c|}{TM} & PV &  &  \\ \hline
			\multicolumn{2}{l|}{TinySSH} & 18K & 0.476K & 136 & 1223 & 204 & 107 & 117 & \multicolumn{1}{c|}{120} & \multicolumn{1}{c|}{493} & \multicolumn{1}{c|}{7.32} & 0.114 & TM \& PV & DHKA, SE, DS, HF\\
			\multicolumn{1}{c|}{\multirow{2}{*}{WG}} & Initiator & \multirow{2}{*}{27K} & \multirow{2}{*}{1.323K} & 68 & 1482 & 260 & \multirow{2}{*}{150} & \multirow{2}{*}{181} & \multicolumn{1}{c|}{60} & \multicolumn{1}{c|}{67} & \multicolumn{1}{c|}{\multirow{2}{*}{1.28}} & \multirow{2}{*}{13.266} & \multirow{2}{*}{TM \& PV} & \multirow{2}{*}{DHKA, AEAD, HF} \\
			\multicolumn{1}{c|}{} & Responder &  &  & 153 & 1389 & 380 &  &  & \multicolumn{1}{c|}{60} & \multicolumn{1}{c|}{50} & \multicolumn{1}{c|}{} &  &  & 
		\end{tabular}%
	}
	\caption{Case studies. Abbreviations used: WG (WireGuard), DHKA (Diffie-Hellman Key Agreement), SE (Symmetric Encryption), DS (Digital Signatures), HF (Hash Functions), and AEAD (Authenticated Encryption with Additional Data).
            \revisedtxt{We report the runtime for 
                 preprocessing and symbolic execution ($\sbirsymb$), 
                construction of the symbolic tree plus model extraction ($\SapicOld$),
                and  for
            verification using \Tamarin (TM) and \ProVerif (PV).}
        }
	\label{fig:testcases}
\end{table*}

\subsection{Verification of TinySSH and WireGuard}\label{evaluation}
We have verified the TinySSH and WireGuard protocols to evaluate our framework.
\newtext{Our case studies demonstrate that our methodology does not introduce any artifacts that inhibit verification.}
\autoref{fig:testcases} shows data we have collected during our evaluation.
The LOC of ARM assembly represents the complete assembly code, including crypto functions, which were necessary to consider in our preprocessing step to compute the program's control flow.

\revised{\infruleref{Rc-w5}{RC,W5}}{The binary of our case studies is unaltered; however, the verifier must manually initialize and steer the verification process. 
Specifically, the user is required to specify: 
(a) to the lifter, the code fragments to be analyzed,
(b) to the symbolic execution engine, which operates on the output of the lifter, i.e., \birsymb{} code, the function names grouped as trusted (libraries) or untrusted (network), 
(c) to the symbolic execution engine, the symbolic model of the cryptographic functions, and 
(d) the assumptions regarding the crypto primitives and the security properties we proved for our case studies in the \Sapic input file.
}

TinySSH is a minimalistic SSH server that implements a subset of SSHv2 features and ships with its own crypto library.
To establish authentication requirements for any parties connecting to TinySSH, we used \SapicOld to model the client of the SSH protocol.
We also automatically extracted the \newrevised{server} model of TinySSH from its ARMv8 machine code, hence covering a system composed of three components written in three very different languages, in ARMv8, \SapicOld and DY rules.
We verified mutual authentication~\cite{lowe1997hierarchy} and forward secrecy~\cite{cohn2016post} with \ProVerif and \Tamarin.

WireGuard implements virtual private networks akin to IPSec and OpenVPN. 
It is quite recent and was incorporated into the Linux kernel.
\revised{\infruleref{Rc-w5}{RC,W5}}{We focused on the handshake protocol of WireGuard instead of the record protocol, as the handshake is usually considered more challenging.}
We have extracted, \emph{for the first time}, the $ \SapicOld $ model of the Linux kernel's WireGuard implementation binary.
Our model is more faithful than existing manual models which, for
instance, use pattern matching for authentication verification
and was extracted automatically.
Having extracted the handshake and the first message transmitted upon the completion of exchanging keys, we prove that the protocol participants mutually agree on the resulting keys in both \ProVerif and \Tamarin.
Moreover, we show the resulting keys remain unknown to the attacker by proving the forward secrecy property using \ProVerif and \Tamarin.

\revised{\infruleref{Ra-w2}{RA,W2}}{
We employed the combined deduction relations $ \dedrelcombSbirDY $ and $ \dedrelcombSapicDY $ in our case studies and \newrevised{extracted a model from the $ \birsymb $ program in}~\Cref{ex:masked-encryption} using our toolchain to demonstrate their application.
\newrevised{
The models reflect $ \dedrelcombSbirDY $ and $ \dedrelcombSapicDY $ as destructors defined in the translation from \SapicOld (and Co.) to \Sapic.
These destructors derive the same terms that $ \dedrelcombSbirDY $ derives in~\autoref{fig:bir-exp4}.
}
Additionally, as we extracted formal models of TinySSH and WireGuard from their respective implementations, 
we have identified no instances in which the DY attacker could acquire additional knowledge through the use of these combined deduction relations. 
}

\ProVerif and \Tamarin exhibit significantly different verification
times. 
For WireGuard, \Tamarin verifies our properties in 1.28 seconds, while \ProVerif takes 13.266 seconds.
Conversely, for TinySSH, \ProVerif outperformed \Tamarin, completing the verification task in 0.114 seconds compared to \Tamarin's 7.32 seconds.

\section{Related Work}
\label{sec:relatedwork}

\begin{table}[t]
\adjustbox{varwidth=\linewidth,scale=0.80}{%
		\begin{tabular}{@{} cl*{5}c @{}}
                        & {\multirow{2}{*}{\begin{tabular}[c]{@{}c@{}} \textbf{Papers}\end{tabular}}} & {\multirow{2}{*}{\begin{tabular}[c]{@{}c@{}}Model\\Origin\end{tabular}}} & {\multirow{2}{*}{\begin{tabular}[c]{@{}c@{}}Attacker\\Model\end{tabular}}} & {\multirow{2}{*}{\begin{tabular}[c]{@{}c@{}}No Parsing\\Assum.\end{tabular}}} &  {\multirow{2}{*}{\begin{tabular}[c]{@{}c@{}} Formalized\end{tabular}}}\\ \\
			\cmidrule{2-6}
                        & Sprenger et al.~\cite{sprengerIglooSoundlyLinking2020} & Required &  DY &   \NO  & Isabelle/HOL  \\     
                        & Arquint et al.~\cite{arquint2022sound,arquint2023generic} & Required &  DY & \NO & ---  \\     
                        & Hahn et al.~\cite{hahn2013compositional} & Required & Comp. & \OK & ---  \\     
                        & Sammler et al.~\cite{sammlerDimSumDecentralizedApproach2023} & Required &  Unbounded & \NO  & Coq  \\   
                        & Bhargavan et al.~\cite{dystareurosp2021} &
                                Code-based &  DY  & \OK & $F^*$  \\ 
                        & Wallez et al.~\cite{wallezComparseProvablySecure2023} &
                                Code-based &  DY  & \OK & $F^*$  \\     
                        & Bhargavan et al.~\cite{bhargavan2008verified,bhargavan2008cryptographically} & Extracted &  DY / Comp. & \NO & ---  \\   
                        & Aizatulin et al.~\cite{aizatulinVerifyingCryptographicSecurity2015} & Extracted &  DY / Comp. & \NO & ---  \\   
                        & Nasrabadi et al.~\cite{nasrabadi2023cryptobap} & Extracted &  DY / Comp. & \NO & ---  \\   
                        & \textbf{This work} & Extracted &  DY & \OK & HOL4 \\         
			\cmidrule[1pt]{2-6}
		\end{tabular}
	}
        \caption{Selected approaches; Comp = Computational, No Parsing Assum. = No strict parsing assumptions (see \autoref{subsec:mess-DY})
        \label{tab:rel-work}}
\end{table}

In recent years, several techniques for verifying crypto protocol implementations have emerged. We survey those based on separation logic, model validation, wrappers, and the CompCert framework~\cite{leroy2006formal}, in this section.
\autoref{tab:rel-work} compares selected works and our proposed approach.

\paragraph{\bf  Separation logic}
~\cite{sergey2017programming,koh2019c,oortwijn2019practical,sprengerIglooSoundlyLinking2020,arquint2022sound,arquint2023generic} used separation logic to analyze the implementation of security protocols.
Sprenger et al.~\cite{sprengerIglooSoundlyLinking2020} introduced a methodology where a protocol model is first formalized in Isabelle/HOL~\cite{paulson2013proof} and then translated into I/O specifications, which are verified using separation-logic based verifiers.
Arquint et al.~\cite{arquint2022sound} extended this to the \Tamarin{} verifier to enable verification against \Tamarin{'s} models.  
Follow-up work~\cite{arquint2023generic} stepped away from
verifying against the specification, and directly verified the protocol properties, which are \emph{stable under concurrency}, by building on a programming language that incorporates protocol operations and modeling the attacker in that language. 

Others~\cite{sprengerIglooSoundlyLinking2020,arquint2022sound,arquint2023generic} used verifiers like Nagini~\cite{eilers2018nagini} for Python, Gobra~\cite{wolf2021gobra} for Go, and VeriFast~\cite{jacobs2011verifast} for Java and C.
Nonetheless, the soundness of these approaches
depends on the correctness of utilized verifiers, including their
dependencies (e.g., Nagini relies on
Viper~\cite{muller2016viper}). In theory, they could be proven
sound, but the languages are not ideal for formalized results.
By contrast, $ \birsymb $'s decompilation approach already provides a formalized soundness results 
from lifting to symbolic execution while covering machine code produced by compiling these languages.

Throughout this line of work~\cite{sprengerIglooSoundlyLinking2020,arquint2022sound,arquint2023generic},
a translation function maps between byte-level messages and DY terms
(or an injective function in the other direction).
Therefore, our arguments in~\autoref{subsec:mess-DY} apply. However, a cursory
glance suggests that our framework (i.e., the subject of \autoref{sec:parallel}) 
might help with this, as their proof structure
likewise consists of a refinement step, followed by decomposition and translation to
a verification language and their communication model builds on
a (subclass of) LTS and CSP-style parallel composition (with built-in
translation).

\paragraph{\bf  Model validation}
Several formalisms were proposed for modeling distributed systems~\cite{brinksma2005process,meseguer2006specification,strubbe2006compositional,hahn2013compositional}, including a hierarchical modeling language and the hybrid process calculus~\cite{brinksma2005process} that focuses on bisimulation notions and congruence results w.r.t.\ parallel composition.
Strubbe et al.~\cite{strubbe2006compositional} introduced a technique to deal with 
the nondeterminism in distributed systems, which was later extended by Meseguer et al.~\cite{meseguer2006specification} to handle the asynchrony of communications.

The methodologies in~\cite{meseguer2006specification,strubbe2006compositional,hahn2013compositional} required checking cross-system variable consistency during communication due to shared variables. This direct impact of one component's actions on others poses a challenge.
In contrast, our synchronization method relies on events containing symbols. By avoiding the reuse of symbols, cross-system consistency is not a concern for us. This improves \newtext{our approach's efficacy} and makes it ideal for modeling distributed systems.

\paragraph{\bf Model extraction}

The application of our theory builds on \CryptoBap~\cite{nasrabadi2023cryptobap},
which derives the idea of extracting protocol models via symbolic execution from 
Aizatulin et al.~\cite{aizatulinExtractingVerifyingCryptographic2011,aizatulinComputationalVerificationProtocol2012}.
Both approach
build upon computational soundness, which 
imposes stringent requirements on the use of cryptography and protocols.
Computational soundness is incredibly difficult to prove
mechanically~\cite{lochbihler2016probabilistic}, which was the
main motivation for our framework,
as it 
(a) enables us to avoid the detour via computational
soundness and
(b) enables compositional proofs.
Where
both \cite{nasrabadi2023cryptobap,
aizatulinVerifyingCryptographicSecurity2015} rely on pen-and-paper
proofs in the cryptographic model,
we have a mechanized end-to-end
proof.

\paragraph{\bf DY code analysis}
\revised{\infruleref{Rc-w2}{RC,W2}}{
Similar to our case studies, $\mathit{DY}^*$~\cite{dystareurosp2021} 
permits code analysis w.r.t.\ a DY attacker, but for a high-level
language ($F^*$) that allows conducting proofs using dependent types.
Their DY attacker is formulated within $F^*$, whereas our framework
results apply to a DY attacker that may compose with other languages. 
Proofs in their framework are internal to $F^*$, while we depend on the
correctness of the protocol verifiers.
\cite[Sec.\ 1]{dystareurosp2021} discusses the trade-off in
automation versus modularity.
$\mathit{DY}^*$ is complemented by Comparse~\cite{ wallezComparseProvablySecure2023}
which provides type combiners and lemmas to deal with packing and
unpacking.
These lemmas are proven at the bit-level, solving (in many cases) the
problems of \emph{limited bit-level reasoning} and \emph{strong parsing assumptions}
mentioned in~\autoref{sec:introduction} ---although we emphasize that
$\mathit{DY}^*$ and Comparse are not a multi-language composition, instead
integrating the DY attacker as a library.
Instead of constructing the format types (and their proofs of
validity), we extract the formats using our symbolic execution as
\birsymb{}-level terms. We translate those to DY terms, possibly losing
bit-level message confusing attacks should the deduction combiner be
incomplete. 
The criteria in \cite[Sec.\ 2]{wallezComparseProvablySecure2023} could
be useful to judge the soundness of this deduction combiner w.r.t.\
the message formats that are abstracted in this way, i.e., w.r.t.\
a given (set of) implementations. 
}

\paragraph{\bf Wrappers}
Research on multi-language semantics has explored translation between languages using a wrapper~\cite{matthews2007operational,ahmed2011equivalence,new2018graduality,sammlerDimSumDecentralizedApproach2023}.
DimSum~\cite{sammlerDimSumDecentralizedApproach2023}
is the most relevant to our work.
DimSum's wrapper-based composition
($ \lceil \cdot \rceil_{_{ \fstlangcol{1} \rightleftharpoons \sndlangcol{2} }}$)
serves as a translation tool between two components written in different
languages as well as between a component and the environment.
Like Igloo~\cite{sprengerIglooSoundlyLinking2020},
they reason about an arbitrary number of languages communicating via events
and build on CSP-style parallel composition and translate
between languages. Instead, 
we use a shared set of symbols to denote equations and deduce
relations between bitstrings in different languages.
DimSum requires $ m^2 $ wrappers to facilitate communication of $m$ languages, suffering from a complexity blow-up associated with compositional
soundness.
Our generic deduction combiners (\autoref{subsec:DedComb}) can remove
this burden.
For computational attackers, 
DimSum's composition does not support probabilistic semantics
and it lacks a notion of runtime bounds for attackers.
As far as the DY model goes, the issues in \autoref{subsec:mess-DY} apply (e.g., there is no single suitable DY term that
 $\lceil \lceil   \senc(m,k)\rceil_{_{\Term \rightharpoonup  \fstlangcol{\imlvals}}} \ \fstlangcol{+ \ 0x1} \rceil_{_{\Term \leftharpoondown  \fstlangcol{\imlvals}}}$ should give). 

\paragraph{\bf CompCert}
CompCert was also used to verify the multi-language protocols at the assembly-code level~\cite{ramananandro2015compositional,stewart2015compositional,gu2015deep,gu2018certified,wang2019abstract,song2019compcertm,koenig2021compcerto,oliveira2022layered}. Among others,~\cite{ramananandro2015compositional,stewart2015compositional,song2019compcertm,koenig2021compcerto} achieved multi-language composition by enforcing a common interaction protocol across all languages,
while~\cite{gu2015deep,gu2018certified,wang2019abstract,oliveira2022layered} enforce specific memory-sharing patterns, along with other restrictions, on the interaction between different components.
In contrast, we neither depend on a common language nor impose any restrictions on the interaction of components.
Our model uses symbols for communication and predicates over these symbols for reasoning, allowing the verification toolchain
to understand this interaction.

\section{Concluding Remarks}
\label{sec:conclusion}

We proposed a framework for symbolic parallel composition that enables composing components operating on different atomic types. Our approach extends the state-of-the-art composition techniques, allowing efficient handling of cross-language communication. Notably, our approach avoids the need to translate incompatible base values and offers a more versatile and applicable solution. By using symbolic values for communication, our method addresses the mismatches encountered in previous translation-based approaches. This provides a more accurate representation of DY terms as symbolic abstractions. 
Our composition framework is multi-language in this first sense: 
our WireGuard case study, for instance, combines programs in the $\SapicOld$, $\sbirsymb$, and DY language in the same system (e.g., \refappendixorfull{eq:multi} in~\theappendixorextended{sec:multi-party}).
Our case studies are also multi-language in a much more pragmatic sense:
any language that compiles to a supported assembly language is
supported, independent of the compiler, and whether it is correct.

In the future, we aim to extend our framework with probabilistic reasoning. We will extend our semantic configuration to include the probability of reaching a given state. 
This will allow us to reason probabilistically about the composition of non-probabilistic languages.

\subsection*{Acknowledgment}
We thank anonymous reviewers for their insightful comments. This work was partially supported by the Wallenberg AI, Autonomous Systems and Software Program (WASP) funded by the Knut and Alice Wallenberg Foundation. We also gratefully acknowledge a gift from Intel and Amazon.

\bibliographystyle{IEEEtran}
\bibliography{refs, zotero}

\begin{full}
\appendix
\section{Appendix}

\section{Supplemental Material}

\subsection{Partially Synchronized Interleaving on Traces}
\label{sec:interleave}
Partially synchronized interleaving on traces generalizes interleaving composition by requiring certain actions from two or more components to occur in a specific relative order that signifies synchronization points. Conversely, other actions remain unconstrained and may be interleaved in any arbitrary manner.

\begin{definition}[Partially Synchronized Interleaving on Traces]\label{def:interleaving}
    	For any LTS $\fstlangcol{M}$ and $ \sndlangcol{M}$, and
    	two sets of traces produced by these LTS, respectively, $ \fstlangcol{\Tracesfun{M}} $ and $ \sndlangcol{\Tracesfun{M}} $,
    	the Partially Synchronized Interleaving on Traces $  \interleaving{\fstlangcol{\Tracesfun{M}}}{ \sndlangcol{\Tracesfun{M}}} $ is the set of all possible traces $ \InterTraces $ such that:
\begin{itemize}
	\item $ \InterTraces $ is a permutation of $  \fstlangcol{\Tracesfun{M}} \cup \sndlangcol{\Tracesfun{M}} $.
	\item The relative order of elements in $  \fstlangcol{\Tracesfun{M}}$ and $\sndlangcol{\Tracesfun{M}} $ are preserved in $ \InterTraces $:
	\begin{itemize}
		\item For all traces $ \fstlangcol{\Trace} \in \fstlangcol{\Tracesfun{M}} $ and $ \InterTrace \in \InterTraces$, 
		$ i$, $j$, $m$, and $n $ such that $ 0 \leq i < j  < m $, 
		there exist $ k$ and $ l $ such that $ 0 \leq k < l  < m + n$,
		$ \InterTrace[k] = \fstlangcol{\Trace[i]} $  and
		$ \InterTrace[l] = \fstlangcol{\Trace[j]} $.
		\item For all traces $ \sndlangcol{\Trace} \in \sndlangcol{\Tracesfun{M}} $ and $ \InterTrace \in \InterTraces$, 
		$ x$, $y$, $m$, and $n $ such that $ 0 \leq x < y  < n $, 
		there exist $ z$ and $ d $ such that $ 0 \leq z < d  < m + n$,
		$ \InterTrace[z] = \sndlangcol{\Trace[x]} $  and
		$ \InterTrace[d] = \sndlangcol{\Trace[y]} $.
	\end{itemize}
	\item For all traces $ \fstlangcol{\Trace} \in \fstlangcol{\Tracesfun{M}} $, $ \sndlangcol{\Trace} \in \sndlangcol{\Tracesfun{M}} $ and $ \InterTrace \in \InterTraces$, 
	$ i$, $j$, $m$, and $n $ such that $ 0 \leq i < m $, $ 0 \leq j < n $, and
	$  \fstlangcol{\Trace[i]} = \sndlangcol{\Trace[j]} $, 
	there exists a $k$ such that $ 0 \leq k  < m + n$ and
	$ \InterTrace[k] = \fstlangcol{\Trace[i]} = \sndlangcol{\Trace[j]}  $.
\end{itemize}
\end{definition}

\subsection{Transitions (De-)Activation}
\label{sec:transitions}
\autoref{def:transition-enabling} defines when adding a predicate can activate a transition in our system.
\begin{definition}[Transition Enabling]\label{def:transition-enabling}
    Given two symbolic LTS
$S_i = (\symbSpace,\configSpace_i,\eventSpace_i,\transrel{}{i}{}{},\predSpace_i,\dedrel_i)$, 
$i\in\set{\colone,\coltwo}$,
their symbolic parallel composition
$\Parallel[\dedrelcomb]{\SysFst[ _1 ]}{\SysSnd[ _2 ]} =
(\symbSpace, 
\subone{\configSpace}\times \subtwo{\configSpace},
\subone{\eventSpace} \cup \subtwo{\eventSpace}, 
\transrel{}{\colonetwo}{}{},
\subone{\predSpace} \uplus \subtwo{\predSpace}, \dedrelcomb)$,
a predicate set $ \comppredset \in 2^{(\subone{\predSpace} \uplus \subtwo{\predSpace})} $
and a predicate $ \predcomb \in (\subone{\predSpace} \uplus \subtwo{\predSpace}) $, 
such that $ \comppredset \dedrelcomb \predcomb$,
we say the predicate $ \predcomb$ enables the transition $\transrel{}{\colonetwo}{}{}$
if:
\begin{itemize}
	\item Either {\small$\transrel{\subone\eventsym}{\colonetwo}{\tuple{\symbset,\comppredset\cup\set{\predcomb},\subone{\config},\subtwo{\config}}}{\tuple{\symbset',\comppredset'\cup\set{\predcomb},\subone{\config'},\subtwo{\config}}}$} or {\small$\transrel{\subtwo\eventsym}{\colonetwo}{\tuple{\symbset,\comppredset\cup\set{\predcomb},\subone{\config},\subtwo{\config}}}{\tuple{\symbset',\comppredset'\cup\set{\predcomb},\subone{\config},\subtwo{\config'}}}$}, 
	and, without adding the predicate $ \predcomb $, 
	it is not possible to move with 
	$ \eventsym_i \in \eventSpace_i \setminus (\subone{\eventSpace} \cap \subtwo{\eventSpace})$ for $i \in\set{\colone,\coltwo}$, i.e.,
	$ \ntransrel{\eventsym_i \ }{i}{\tuple{\symbset,(\proj{i}{\comppredset}),\config_{i}}}{\tuple{\symbset',(\proj{i}{\comppredset'}),\config'_{i}}}$,
	keeping the 
	complement's
	predicate set untouched $ \proj{\overline i}{\comppredset} = \proj{\overline i}{\comppredset'} $,
	\item Or {\small $\transrel{\eventsym}{\colonetwo}{\tuple{\symbset,\comppredset\cup\set{\predcomb}, \subone{\config},\!\subtwo{\config}}}{\tuple{\symbset'\!,\comppredset'\cup\set{\predcomb},\subone{\config'},\subtwo{\config'}}}$}, and, $ \transrel{\eventsym}{i}{\tuple{\symbset,(\proj{i}{\comppredset}),\config_{i}}}{\tuple{\symbset'_{i},(\proj{i}{\comppredset'}),\config'_{i}}} $ is not possible without adding the predicate $ \predcomb $, for $i\in\set{\colone,\coltwo}$, $ \eventsym \in \subone{\eventSpace} \cap \subtwo{\eventSpace} $, and $ \symbset' = \subone{\symbset'} \cup \subtwo{\symbset'} $.
\end{itemize}
\end{definition}

Adding predicates may also disable transitions within the system. The definition for when adding a predicate disables transitions is similar to~\autoref{def:transition-enabling} and obtained by negating logical entailment.
\subsection{Composing and Decomposing DY
	libraries}\label{sec:dy-lib-comp}

Protocol parties are often implemented in different languages that potentially incorporate different implementations of the same cryptographic library. Additionally, each party may employ additional libraries tailored to their specific needs, which could differ from those used by others. Therefore, our framework needs to account for both scenarios in the composition of protocol participants.
%
We use function symbols, which represent cryptographic operations, to distinguish between the two scenarios where DY libraries have identical or distinct function symbols. 
We introduce the following corollary---and mechanize its proof in HOL4 to enable the composition or decomposition of DY libraries. See~\href{https://github.com/FMSecure/CryptoBAP/tree/main/HolBA/src/tools/parallelcomposition/DYLib/derived_rules_DYlibScript.sml#L25}{\itshape\underline{Same-Signature}} and~\href{https://github.com/FMSecure/CryptoBAP/tree/main/HolBA/src/tools/parallelcomposition/DYLib/derived_rules_DYlibScript.sml#L63}{\itshape\underline{Distinct-Signatures}} for the proof of~\autoref{lem:samesig}. 

\begin{corollary}
	\label{lem:samesig}
	For
	all DY libraries \fstlangcol{\DYLib{\subone\Funcs}} and \sndlangcol{\DYLib{\subtwo\Funcs}}, where $\subone\Funcs$ and $\subtwo\Funcs$
	can be the same or distinct function signatures, we have that
	$ \sbirmixtraces(\SParallel[]{\fstlangcol{\DYLib{\subone\Funcs}}}{\sndlangcol{\DYLib{\subtwo\Funcs}}}) = \interleaving{\fstlangcol{\Traces}(\fstlangcol{\DYLib{\subone\Funcs}})}{\sndlangcol{\Traces}(\sndlangcol{\DYLib{\subtwo\Funcs}})} $.
\end{corollary}

\autoref{lem:samesig} serves not only in the composition but also in the decomposition of a single DY library. This allows us to break down a DY library, containing function symbols, into the composition of two DY libraries, each with either the exact same signature or distinct signatures. Consequently, each protocol participant's library can be decomposed into two parts, such as $\SParallel[]{{\DYLib{\Funcs}}}{\fstlangcol{\DYLib{\fstlangcol\Funcs}}}$ or $ \SParallel[]{{\DYLib{\Funcs}}}{\sndlangcol{\DYLib{\sndlangcol\Funcs}}}$.

Following this line of reasoning, when composing multiple parties, it becomes possible to independently compose each part of each participant's library (i.e., $ \SParallel[]{{\DYLib{\Funcs}}}{{\DYLib{\Funcs}}}$ for the common and $ \SParallel[]{\fstlangcol{\DYLib{\fstlangcol\Funcs}}}{\sndlangcol{\DYLib{\sndlangcol\Funcs}}}$ for the remainder). 
Now the common part can be merged into one ($ \DYLib{\Funcs} $).
For more details about the application of~\autoref{lem:samesig} in one of our case studies, see Appendix~\ref{sec:multi-party}.

\subsection{Refinement of Composed System}\label{sec:comp-refine}

A refinement relation $ \sqsubseteq $ relies on an interpretation function that maps symbolic variables to concrete values and can be applied to a symbolic trace. 
\autoref{subsec:refinement} explains this process using notations like $\fstlangcol{\interpreterapply}$ and $\sndlangcol{\interpreterapply}$. 

We define how this refinement applies to the composed system.
Let $ \ \traceproj{i}{} \ : 2^{\subone{\eventSpace} \cup \subtwo{\eventSpace}} \to 2^{\eventSpace_i} $ denotes the trace projection to $i\in\set{\colone,\coltwo}$, then:

\begin{definition}[Composed System Refinement]\label{def:ref-comp}
	Let $ \conmixtrace $ be a concrete composed trace and $ \sbirmixtrace $ be a symbolic composed trace, then a refinement relation between these two traces is $ \conmixtrace \sqsubseteq \sbirmixtrace$ such that
	there exist interpretation functions $ \fstlangcol{\interpret_{_1}} $ and $ \sndlangcol{\interpret_{_2}} $ where
	\begin{itemize}
		\item $ \traceproj{\fstlangcol{1}}{\conmixtrace} \fstlangcol{ =  \interpreterapply(}\traceproj{\fstlangcol{1}}{\sbirmixtrace}\fstlangcol{,\interpret_{_1})} $
		\item $ \traceproj{\sndlangcol{2}}{\conmixtrace} \sndlangcol{ =  \interpreterapply(}\traceproj{\sndlangcol{2}}{\sbirmixtrace}\sndlangcol{,\interpret_{_2})} $
		\end{itemize}
\end{definition}

\subsection{Concrete World}

\label{sec:conc-thms}
In the CSP-style parallel composition of concrete labeled transition systems, synchronization and communication enable interaction among sub-components in a composed system.
A correspondence can be established between traces of a composed system using CSP-style asynchronous parallel composition and the interleaving of traces of each sub-component.
\begin{theorem}[Concrete Composition Correctness]
	\label{thm:conccompthm}
	For any concrete LTS $\fstlangcol{M}$ and $ \sndlangcol{M} $, we have $ \conmixtraces(\CParallel{\fstlangcol{M}}{\sndlangcol{M}}) = \interleaving{\fstlangcol{\contraces}(\fstlangcol{M})}{\sndlangcol{\contraces}(\sndlangcol{M})} $.	
\end{theorem}
\begin{proof}
The goal is to show that for all traces of the composition of concrete LTS, there is an equivalent trace resulting from interleaving the traces of each concrete LTS and vice versa. We prove the theorem using induction over the length of the composed traces. Considering no steps were undertaken, the base case is straightforward. For the inductive case, we utilize case distinction over synchronous and asynchronous events.
\end{proof}
\autoref{thm:conccompthm} enables compositional analysis, as evidenced by the following corollary, wherein individual components can be refined while preserving trace inclusion for the composed system.
\begin{corollary}[Concrete Compositional Trace Inclusion]
	\label{thm:con-comprefin}
	For any concrete LTS $\fstlangcol{M_1}$, $ \fstlangcol{M_2} $, $ \sndlangcol{M_1} $, and $ \sndlangcol{M_2} $, we have
	\begin{center}
		\begin{prooftree}	
			\hypo{\begin{matrix}
					\fstlangcol{\contraces}(\fstlangcol{M_1}) \subseteq \fstlangcol{\contraces}(\fstlangcol{M_2})
			\end{matrix}}
			\hypo{\begin{matrix}
					\sndlangcol{\contraces}(\sndlangcol{M_1}) \subseteq \sndlangcol{\contraces}(\sndlangcol{M_2})
			\end{matrix}}
			\infer2[]{  \conmixtraces(\CParallel{\fstlangcol{M_1}}{\sndlangcol{M_1}})  \subseteq \conmixtraces(\CParallel{\fstlangcol{M_2}}{\sndlangcol{M_2}})}
		\end{prooftree}
	\end{center}
\end{corollary}
Similar results were previously established for CSP-style asynchronous parallel composition of concrete systems (see, e.g.,~\cite{silva2012shared}), but we have formalized and proven them on top of HOL4. 
Complete mechanized proofs are available at~\href{https://github.com/FMSecure/CryptoBAP/tree/main/HolBA/src/tools/parallelcomposition/concrete/interleavingconcreteScript.sml#L234}{\itshape\underline{Concrete-Composition}}  and \href{https://github.com/FMSecure/CryptoBAP/tree/main/HolBA/src/tools/parallelcomposition/concrete/interleavingconcreteScript.sml#L246}{\itshape\underline{Concrete-Trace-Inclusion}}.
\subsection{Relationship to Computational Soundness}
\label{rel-comp-sound}
Computational soundness says that any computational trace 
(i.e., a trace produced by the protocol implementation and some probabilistic polynomial-time (PPT) Turing machine)
is either improbable or an instance of a symbolic trace with a DY
attacker.
Cortier and Warinschi found that computational soundness can be
obtained from two conditions: deduction soundness and the commutation property~\cite{cortierComposableComputationalSoundness2011}.
The appeal of this approach is that deduction soundness is somewhat
composable~\cite{cortierComposableComputationalSoundness2011,bohlDeductionSoundnessProve2013a}, while computational soundness is not (as far as we know),
although  deduction soundness without the commutation property provides only guarantees against 
passive attackers.

Assumption~2 in~\autoref{sec:end-to-end} seems to be conceptually close to deduction soundness.\footnote{Traditionally,
	computational soundness and related notions hardcode the
	complexity-theoretic execution model, so we have to argue the
	equivalence in spirit. Deduction soundness says that the
	computational attacker
	is unlikely to produce a bitstring that can be parsed to a DY term
	that is undeducible based on the terms received so far.
	Indeed, any such bitstring would result from a computational trace
	that could not be described as a refinement of some (symbolic) trace from 
	$\SParallel[{\dedrelAvsL}]{L^{DY}}{A^{DY}}$. Vice versa, any concrete trace from 
	$\CParallel{L^{\birsymb}}{A}$ that is not an instance of a symbolic trace
	must either be due to an incorrect library implementation or due to
	$A$, in which case it constitutes an `undeducible' bitstring. 
	A formal argument would require a probabilistic notion of refinement,
	but constitutes an interesting pursuit.
}
This indicates that the commutation property may be an artifact of the
translation approach, and not in fact necessary to achieve the aims of
computational soundness (thus opening up the possibility of
composition results like for deduction soundness).
Roughly speaking, the commutation property states that no
(concrete) PPT Turing machine can distinguish between the
concrete (computational) protocol and a translation function around the
DY interpretation of the protocol. The step using
\cite[Thm.~4.1]{lindner2023proofproducingsymbolicexecutionbinary} fulfills
the same purpose, but there is no translation inside the system;
instead, the instantiation is a meta-mathematical relation between the
traces of the concrete system and the symbolic system.
The difference becomes tangible when considering the proof effort.
In a proof like \cite[Thm.~4.1]{lindner2023proofproducingsymbolicexecutionbinary}, the
researcher is given a concrete trace and can provide a mapping on the
spot, as long as they can justify the symbolic trace the mapping
applies to. For the commutation property, the researcher has to provide
a PPT algorithm that not
only translate
every single concrete trace, but is also reversible.
We, therefore, think that this assumption merits deeper exploration.

\subsection{Multi-Party Proof Structure}
\label{sec:multi-party}
In this section, we elucidate our proof structure for the composition of multiple protocol participants, cryptographic libraries, and an unspecified attacker $ A $.
Consider the ARMv8 programs corresponding to the WireGuard initiator ($ I^{ARM} $) and responder ($ R^{ARM} $), along with their employed cryptographic libraries ($ {L_{i}}^{ARM} $ and $ {L_{r}}^{ARM} $ respectively).
\begin{align}
    & \contraces (\CParallel{\CParallel{(\CParallel{I^{ARM}}{{L_{i}}^{ARM}})}{(\CParallel{R^{ARM}}{{L_{r}}^{ARM}})}}{A})\label{eq:first}
\\	& = \contraces (\CParallel{\CParallel{(\CParallel{I^{\birsymb}}{{L_{i}}^{\birsymb}})}{(\CParallel{R^{\birsymb}}{{L_{r}}^{\birsymb}})}}{A}) 
	\intertext{By employing \cite{DBLP:journals/scp/LindnerGM19}'s lifter, we obtain corresponding $\birsymb$ programs and demonstrate their composition with $ A $ using~\autoref{thm:con-comprefin}.
Building upon the soundness of the symbolic execution engine~\cite[Thm.~4.1]{lindner2023proofproducingsymbolicexecutionbinary} and relying on an assumption about the attacker, as discussed in~\autoref{sec:end-to-end}, we move from the concrete to the symbolic using the refinement theorem  \autoref{thm:traceinc}.}
	& \sqsubseteq \symtraces (\SParallel[\dedrelAvsL]{\SParallel[\dedrelcombSbirDY]{\SParallel{I^{\sbirsymb}}{R^{\sbirsymb}}}{\underbrace{\SParallel[]{{L_{i}}^{DY}}{{L_{r}}^{DY}}}_{\stackrel{\autoref{lem:samesig}}{=}}}}{A^{DY}})
    \\  & = \symtraces (\SParallel[\dedrelAvsL]{\SParallel[\dedrelcombSbirDY]{\SParallel{I^{\sbirsymb}}{R^{\sbirsymb}}}{ \quad \quad L^{DY} \quad \quad}}{A^{DY}})  \label{eq:to-dy}
	\intertext{As $\SParallel[]{}{}$ is associative w.r.t.\ trace
equivalence, we can employ~\autoref{lem:samesig} to demonstrate the composition of $ {L_{i}}^{DY} $ and $ {L_{r}}^{DY} $ libraries---whether with identical or distinct function signatures---is equivalent to a single DY library ($ L^{DY} $) encompassing all these function signatures.
Subsequently, we apply our translation result from $ \sbirsymb $ to \SapicOld (\autoref{thm:traceinc:sbir-sapic}), by leveraging~\autoref{thm:symb-comprefin} presented in~\autoref{sec:correctness}.}
 	& \sqsubseteq \symtraces (\SParallel[\dedrelcombSapicDY]{\SParallel[\dedrelcombSbirDY]{I^{\sbirsymb}}{R^{\SapicOld}}}{\SParallel[\dedrelAvsL]{L^{DY}}{A^{DY}}})
        \label{eq:multi}
        \intertext{We perform symbolic execution and extract the \SapicOld model for each component individually.}
	& \sqsubseteq \symtraces (\SParallel[\dedrelcombSapicDY]{\SParallel{I^{\SapicOld}}{R^{\SapicOld}}}{\SParallel[\dedrelAvsL]{L^{DY}}{A^{DY}}})
        \label{eq:before-merge}
     \\ & =  \symtraces (\SParallel[\dedrelcombSapicDY]{IR^{\SapicOld}}{\SParallel[\dedrelAvsL]{L^{DY}}{A^{DY}}})
	\quad\text{with $IR = \sndlangcol{I \mid R}$}
        \label{eq:merge}
	\intertext{
        As the DY attacker and library are included within the semantics of \Sapic,
        we conclude that:}
	& =  \symtraces (IR^{\Sapic})
\end{align}
We have proved this end-to-end correctness result in HOL4, which you can see \href{https://github.com/FMSecure/CryptoBAP/tree/main/HolBA/src/tools/parallelcomposition/instantiations/end_to_end_correctnessScript.sml#L299}{here}.
\paragraph{\bf Extending to arbitrarily many parties}
This argument can be repeatedly applied to cover
an arbitrary but bounded number of protocol implementations.
Depending on the language, the individual components may support
open-ended loops, hence this bound is on the number of components,
e.g., parties, not sessions.
Let $\mathit{RIR} = \sndlangcol{! I \mid !R}$.

\newcommand{\ntimes}[1]{\underbrace{#1}_{\text{$n$ times}}}
\newcommand{\nmtimes}[1]{\underbrace{#1}_{\text{$n-1$ times}}}

\begin{align*}
	& \ \contraces (\CParallel{
            \CParallel{\underbrace{(\CParallel{I^{ARM}}{{L_{i}}^{ARM}})}_{\text{$n$ times}}
                  }{
                      \underbrace{(\CParallel{R^{ARM}}{{L_{r}}^{ARM}})}_{\text{$n$ times}}
          }}{A})
          \intertext{We inductively apply transformations as in the earlier steps (\ref{eq:first}) 
          and (\ref{eq:before-merge}).  (Note each step is applied
          $n$ times, then the next.) }
 	& \sqsubseteq \symtraces (
        \SParallel[\dedrelcombSapicDY]{
               \SParallel{\ntimes{I^{\SapicOld}}}{\ntimes{R^{\SapicOld}}}
       }{\SParallel[\dedrelAvsL]{L^{DY}}{A^{DY}}}
        )
        \intertext{Following step (\ref{eq:merge}), we can draw
            the first initiator component and the first responder
            component together ($\sndlangcol{I \mid R}$)
        over approximate. }
 	& \sqsubseteq \symtraces (
        \SParallel[\dedrelcombSapicDY]{
           \SParallel[\dedrelcombSapicDY]{
               \SParallel{\nmtimes{I^{\SapicOld}}}{\nmtimes{R^{\SapicOld}}}
       }{\SParallel[\dedrelAvsL]{L^{DY}}{A^{DY}}}}
               {RIR^{\SapicOld}}
        )
        \intertext{
            We can repeat this another $n-1$ times, as
            $\sndlangcol{I \mid R \mid !I \mid !R}$
            is equivalent to  
            $\sndlangcol{RIR}$ in \SapicOld and \Sapic.
        }
& =  \symtraces (\SParallel[\dedrelcombSapicDY]{RIR^{\SapicOld}}{\SParallel[\dedrelAvsL]{L^{DY}}{A^{DY}}})
\\
& =  \symtraces (RIR^{\Sapic})
\end{align*}

\end{full}
\end{document}